\documentclass[11pt]{article}

\usepackage{amsmath,amsfonts,amssymb,amsthm}
\usepackage[margin=1in]{geometry}
\usepackage[colorlinks]{hyperref}
\usepackage{enumitem}
\usepackage{mathtools}
\usepackage{enumitem}
\usepackage{algorithm, caption}
\usepackage[noend]{algpseudocode}
\usepackage{subfigure}
\usepackage{graphicx}
\usepackage{tabularx,environ}
\usepackage[braket]{qcircuit}
\usepackage[dvipsnames]{xcolor}
\usepackage{multirow}
\usepackage[capitalize]{cleveref}
\usepackage{dsfont}

\makeatletter
\newcommand*\bigcdot{\mathpalette\bigcdot@{.5}}
\newcommand*\bigcdot@[2]{\mathbin{\vcenter{\hbox{\scalebox{#2}{$\m@th#1\bullet$}}}}}
\makeatother

    \hypersetup{
    	colorlinks,
    	linkcolor={red!100!black},
    	citecolor={blue!100!black},
    }
\makeatletter
\newcommand\notsotiny{\@setfontsize\notsotiny\@vipt\@viipt}
\makeatother

\newcommand{\calP}{\mathcal{P}}

\newcommand{\norm}[1]{\left\|{#1}\right\|}

\def\01{\{0,1\}}

\newcommand{\eps}{\varepsilon}

\newcommand{\BH}{\textsf{BH}}
\newcommand{\CJ}{\textsf{CJ}}

\newcommand{\node}{\mathrm{node}}

\theoremstyle{plain}
\newtheorem{dfn}{Definition}
\newtheorem{theorem}{Theorem}

\newtheorem{lemma}[theorem]{Lemma}

\newtheorem{proposition}[theorem]{Proposition}

\newtheorem{fact}[theorem]{Fact}


\definecolor{applegreen}{rgb}{0.0, 0.5, 0.0}

\renewcommand{\Pr}{\mbox{\rm Pr}}
\renewcommand{\Re}{\mbox{\rm Re}}
\renewcommand{\Im}{\mbox{\rm Im}}

\DeclareMathOperator{\cb}{cb}
\renewcommand{\cb}{{\mathrm{cb}}}
   
\DeclareMathOperator{\GK}{G_{\mathbb K}} 
\DeclareMathOperator{\GR}{G_{\mathbb R}} 
\DeclareMathOperator{\GC}{G_{\mathbb C}}

\DeclareMathOperator{\op}{op} 

\newcommand{\pmset}[1]{\{-1,1\}^{#1}} 

\newcommand{\Id}{\ensuremath{\mathop{\rm Id}\nolimits}}

\newcommand{\poly}{\mbox{\rm poly}}

\DeclareMathOperator{\Tr}{Tr}

\DeclareMathOperator{\Diag}{Diag}

\DeclareMathOperator{\supp}{supp}

\newcommand{\ind}[1]{\mathbf{#1}}

\newcommand{\beq}{\begin{equation}}
\newcommand{\eeq}{\end{equation}}
\newcommand{\beqn}{\begin{equation*}}
\newcommand{\eeqn}{\end{equation*}}
\newcommand{\beqr}{\begin{eqnarray}}
\newcommand{\eeqr}{\end{eqnarray}}
\newcommand{\beqrn}{\begin{eqnarray*}}
\newcommand{\eeqrn}{\end{eqnarray*}}
\newcommand{\bmline}{\begin{multline}}
\newcommand{\emline}{\end{multline}}
\newcommand{\bmlinen}{\begin{multline*}}
\newcommand{\emlinen}{\end{multline*}}

\theoremstyle{plain}

\newtheorem{question}[theorem]{Question}

\theoremstyle{definition}

\theoremstyle{remark}
\newtheorem{remark}[theorem]{Remark}

\renewenvironment{proof}[1][]{
    \begin{trivlist}
    \item[\hspace{\labelsep}{\em\noindent Proof#1:\/}]}
    {{\hfill$\Box$}
    \end{trivlist}
}

\newtheoremstyle{named}{}{}{\itshape}{}{\bfseries}{.}{.5em}{\thmnote{#3}}
\theoremstyle{named}
\newtheorem*{namedtheorem}{Theorem}

\usepackage{tikz}
\usepackage{thm-restate} 

\begin{document}

\title{Learning low-degree quantum objects}

\author{ 
 Srinivasan \\Arunachalam\thanks{IBM Quantum. \href{mailto:Srinivasan.Arunachalam@ibm.com}{Srinivasan.Arunachalam@ibm.com}}
 \and
Arkopal\\ Dutt \thanks{IBM Quantum \href{arkopal@ibm.com}{arkopal@ibm.com} } \and 
Francisco \\ Escudero Gutiérrez\thanks{Qusoft and CWI \href{feg@cwi.nl}{feg@cwi.nl} } \and
Carlos\\Palazuelos \thanks{UCM and ICMAT \href{carlospalazuelos@mat.ucm.es}{carlospalazuelos@mat.ucm.es} } \and }

\date{ \today }

\maketitle

\begin{abstract}

We consider the problem of learning low-degree quantum objects up to $\varepsilon$-error in $\ell_2$-distance. We show the following results:
$(i)$ unknown $n$-qubit  degree-$d$ (in the Pauli basis) quantum channels and unitaries can be learned  using  $O(1/\varepsilon^d)$ queries (independent of $n$), $(ii)$ polynomials $p:\pmset{n}\rightarrow [-1,1]$ arising from  $d$-query quantum algorithms can be classically learned from $O((1/\varepsilon)^d\cdot \log n)$ many random examples $(x,p(x))$ (which implies learnability even for $d=O(\log n)$), and $(iii)$ degree-$d$ polynomials $p:\{-1,1\}^n\to [-1,1]$ can be learned through $O_{}(1/\eps^d)$ queries to a quantum unitary $U_p$ that block-encodes  $p$. Our main technical contributions are new Bohnenblust-Hille inequalities for quantum channels and completely bounded~polynomials.

\end{abstract}



\section{Introduction}
Computational learning theory refers to the mathematical framework for understanding  machine learning models and quantifying their complexity. The seminal result of Leslie Valiant~\cite{DBLP:journals/cacm/Valiant84} (who introduced the Probably Approximately Correct (PAC) model) gives a complexity-theoretic definition of what it means for a class of functions $f:\01^n\rightarrow \01$ to be learnable information-theoretically and~computationally. 

One of the foundational results in computational learning theory is the one of Linial, Mansour and Nisan~\cite{linial1993constant} who showed that $\textsf{AC}_0$, or constant-depth $n$-bit classical circuits consisting of AND, OR and NOT gates, can be learned in quasi-polynomial time.  A crucial aspect of their proof is the following structural theorem: if a Boolean function $f$ is computable by~$\textsf{AC}_0$, then $f$ can be approximated by a \emph{low-degree polynomial}. Using this structural property, the learning algorithm  approximates the coefficients of all the low-degree monomials of $f$ in the PAC model and hence approximately learns the unknown $\textsf{AC}_0$ function. Since their work, the notion of  Boolean functions being low-degree themselves or being approximated by low-degree polynomials has been a central technique~\cite{bun2021guest} to obtain new learning algorithms. Furthermore, low-degree approximations have played a significant role  in theoretical computer science topics such as quantum computing, circuit complexity, learning theory and cryptography. 

In the last few years, there have been several works in quantum learning theory where the goal has been to learn an unknown object on a quantum computer given various access models to the unknown object. There have been works that considered learning unknown states, channels, unitaries, Hamiltonians either with our without structure present in them. Motivated by classical learning theory, in this work our primary focus will be on learning objects that have the additional structure of being \emph{low-degree}. Since we consider different objects, when presenting our results we  will make clear, the definition of being low-degree, but the main motivation of this work can be summarized by the following question:

\begin{quote}
\centering
\emph{Can we learn low-degree $n$-qubit quantum objects information-theoretically with complexity that scales only polynomial (or better polylogarithmic) in $n$?} 
\end{quote}
We give a positive answer to this question for channels, unitaries, quantum query algorithms, polynomials accessed through a block-encoding unitary and states. Below we describe our main technical contribution and~results. 

\subsection{Technical contribution: New Bohnenblust-Hille inequalities}
In 1931, Bohnenblust and Hille~\cite{bohnenblust1931absolute} (generalizing the classic theorem of Littlewood~\cite{littlewood1930bounded}) gave a solution to the famous \emph{Bohr strip problem} of Dirichlet series~\cite{bohr1913ueber}. To do that, they showed the following: let $T:([-1,1]^n)^d\rightarrow [-1,1]$ 
be a $d$-tensor specifed by the coefficients $T=(\widehat{T}_{i_1,\ldots,i_d})_{i_1,\ldots,i_d\in [n]}$ (meaning that $T(x_1,\dots,x_d)=\sum_{i_1,\dots,i_d\in [n]}\widehat{T}_{i_1,\ldots,i_d}x_1(i_1)\dots x_d(i_d)$),  then
\begin{equation}\label{eq:BHfortensors}
    \Big(\sum_{i_1,\ldots,i_d=1}^n |\widehat{T}_{i_1,\ldots,i_d}|^{2d/(d+1)}\Big)^{(d+1)/2d}\leq C(d),
\end{equation}
where $C(d)$ is a universal constant independent of $n$.\footnote{We remark that the inequality above is a simplified version of the original  Bohnenblust-Hille inequality and we discuss this in more detail in the preliminaries.}  Their work marked the birth of the Bohnenblust-Hille ($\BH$) inequality, which became a fundamental tool in functional analysis. Despite being studied for over a century, the best known upper bound on $C(d)$ scales polynomially with $d$ \cite{bayart2014bohr}, while the best lower bound is a constant \cite{diniz2014lower}. Closing this gap has been an active area of research in mathematics. It was not until 2011 that Defant et al.~\cite{defant2011bohnenblust} refined the $\BH$ inequality and found a striking application of the $\BH$ inequality: they determined the precise asymptotic behavior of the $n$-dimensional Bohr radius using the $\BH$ inequality. Since then, there has been renewed interest in the $\BH$ inequality and has found several applications in theoretical computer science such as Fourier-Entropy influence conjecture~\cite{arunachalam2021improved}, classical learning theory~\cite{eskenazis2022learning}, non-local games~\cite{montanaro2012some}, and quantum computing~\cite{huang2022learning,volberg2023noncommutative}.

Of particular interest to our work, the $\BH$ inequality recently got the attention of the computer science community when  Eskenazis and Ivanisvili~\cite{eskenazis2022learning} used a version of it to prove a major improvement on the problem of classical learning bounded low-degree multilinear polynomials (which we discuss in detail below).
The key insight of Eskenazis and Ivanisvili was the following: $\BH$ inequalities imply that given a bounded low-degree polynomial, most of its coefficients are quite small, so one does not have to learn them. Since then, multiple extensions of the $\BH$ have been proved and applied to learning quantum  objects~\cite{huang2022learning,slote2023bohnenblust,slote2023noncommutative,volberg2023noncommutative,klein_et_al:LIPIcs.ITCS.2024.69}. In this work we extend the $\BH$ inequality in two ways: 
\begin{enumerate}
    \item We consider a variant of the $\BH$ inequality, that can be regarded as a hybrid between the $\BH$ inequality and the celebrated Grothendieck inequality~\cite{grothendieck1953resume}.  We show that degree-$d$ completely bounded tensors $\widehat T$ (which are known to be the output of $d$-query  quantum  algorithms~\cite{QQA=CBF})~satisfy
  $$
    \Big(\sum_{i_1,\ldots,i_d=1}^n |\widehat T_{i_1,\ldots,i_d}|^{2d/d+1}\Big)^{(d+1)/2d}\leq 1.
    $$
    In other words, we improve the $\BH$ constant for completely bounded tensors from $\poly(d)$ to~$1$. See \cref{theo:BHGf} for a precise~statement. 
    \item  More recently, the works of~\cite{huang2022learning,volberg2023noncommutative} have considered non-commutative variants of the $\BH$ inequality. They showed that the Pauli coefficients of $n$-qubit degree-$d$ (i.e., $d$-local) operators that are bounded in operator norm  can be bounded similar to Eq.~\eqref{eq:BHfortensors}, but with $\exp (d)$ instead of $\poly(d)$.\footnote{We remark that Volberg, Zhang~\cite{volberg2023noncommutative} obtained a constant of $\exp(d)$, while Huang et al.~\cite{huang2022learning} obtained~$d^{O(d)}$.} 
    Here, we prove another non-commutative version of the $\BH$ inequality for quantum channels (we in fact prove a stronger $\BH$ inequality for superoperators that are bounded in $S_1\to S_\infty$ norm, and refer the reader to Theorem~\ref{theo:BHchannels} for more details). 
 In particular, if $\Phi$ is a quantum channel defined as $\Phi(\rho):=\sum_{x,y} \widehat{\Phi}(x,y)\sigma_x \rho \sigma_y$ (where $\sigma_x$ is a Pauli operator and $\widehat{\Phi}(x,y)=0$ if $|x|+|y|>d$), then the Pauli coefficients $\widehat{\Phi}(x,y)$ satisfy  
    \begin{equation}\label{eq:BHchannels}
    \Big(\sum_{x,y}|\widehat\Phi(x,y)|^{2d/(d+1)}\Big)^{(d+1)/2d}\leq \exp(d).
    \end{equation}
    Here $\exp(d)$ might seem too big compared to the $\poly(d)$ of \cref{eq:BHfortensors}, but \cref{eq:BHfortensors} and \cref{eq:BHchannels} are incomparable:  \cref{eq:BHfortensors} corresponds to tensors, which are a very structured class of polynomials, while \cref{eq:BHchannels} is a non-commutative analogue of the $\BH$ inequality for general polynomials, for which the best upper bounds~\cite{defant2019fourier} are superpolynomial in $d$.

    Our $\BH$ inequality for channels is a generalization from operators to superoperators of the mentioned non-commutative $\BH$ inequalities \cite{huang2022learning,volberg2023noncommutative}, as explained in \cref{rem:channelsimpliesunitaries}. 
    \end{enumerate}
These $\BH$ inequalities might be of independent interest both for mathematicians and quantum computing; in this work we crucially use them to analyze our quantum learning algorithms.

\subsection{Result 1: Learning channels}
Learning a quantum process is a fundamental task in quantum computing and this can be modelled as learning an unknown \emph{quantum channel}, also referred to as \emph{quantum process tomography}.  On an experimental level, the dynamics of closed quantum systems can be modeled as a unitary transformation from the initial state to the final state. However, in practice, quantum systems interact with the environment and are generally treated as an open quantum system. To learn the behavior of these open quantum systems, it is convenient to model this map as a quantum channel~\cite{nielsen_chuang}.
Learning an $n$-qubit quantum channel is however challenging and is known to require $\Theta(4^n)$ queries to the channel~\cite{gutoski2014process}. This exponential sample complexity can be drastically improved when prior information on the structure of the channel is available. For example, a recent work of Bao and Yao~\cite{bao2023nearly} considered $k$-junta quantum channels, i.e., $n$-qubit channels that act non-trivially only on at most $k$ of the $n$ (unknown) qubits leaving the rest of the qubits unchanged. These channels were shown to be learnable using $\widetilde\Theta(4^k)$ queries to the channel~\cite{bao2023nearly}. In~\cite{chung2018sample}, it was shown that quantum channels that can be efficiently generated, can be learned efficiently, albeit in the weaker PAC learning model. In this work, we consider learning  $n$-qubit quantum channels which have a Pauli decomposition only involving low-degree Pauli operators. 

\paragraph{Low-degree quantum channels.} To describe our result, we first describe the Pauli analysis for quantum channels.  An $n$-to-$n$ qubit quantum channel $\Phi$ can be expressed as
\begin{equation}
    \Phi(\rho) = \sum_{x,y \in \{0,1,2,3\}^n} \widehat\Phi (x,y) \cdot \sigma_x \rho \sigma_y,
    \label{def:arb_channel}
\end{equation}
where $\sigma_x=\mathop{\otimes}_{i\in [n]}\sigma_{x_i}$ and  $\sigma_i$ for $i\in \{0,1,2,3\}$ are the Pauli matrices  
\begin{equation*}
    \sigma_0=\begin{pmatrix}
        1 & 0\\ 0 & 1
    \end{pmatrix},\ \sigma_1=\begin{pmatrix}
        0 & 1\\ 1 & 0
    \end{pmatrix},\ \sigma_2=\begin{pmatrix}
        0 & -i\\ i & 0
    \end{pmatrix},\ \sigma_3=\begin{pmatrix}
        1 & 0\\ 0 & -1
    \end{pmatrix};
\end{equation*} 
and $\widehat\Phi (x,y)$ are the Pauli coefficients of the channel. 
Given $x\in \{0,1,2,3\}^n$, $|x|$ is the number of  non-zero entries of $x$.  The degree of a channel $\Phi$ is the minimum integer $d$ such that $\widehat{\Phi}(x,y)=0$ if $|x|+|y|>d$. Our first result is an efficient learning algorithm for low-degree channels. The learning model we consider is the same as the recent work of Bao and Yao~\cite{bao2023nearly}. Given a channel~$\Phi$, a learning algorithm is allowed to make \emph{queries to $\Phi$} as follows: it can choose a state $\rho$, feed $\rho$ to the channel to obtain $\Phi(\rho)$ and measure $\Phi(\rho)$ in any basis. From the measurement outcomes, the learner should output a classical description of a superoperator $\widetilde \Phi$ that is close to $\Phi$ in the $\ell_2$-distance defined by the usual inner product for superoperators, i.e.,  $\langle \Phi, \widetilde{\Phi}\rangle=\Tr[J(\Phi) J(\widetilde{\Phi})]/4^n$, where $J(\Phi)$ is the Choi-Jamiolkowski
($\CJ$) representation of $\Phi$.\footnote{ Throughout this paper we say an algorithm $(\varepsilon,\delta)$-learns a quantum object if it succeeds with probability $\geq 1-\delta$ and outputs an $\varepsilon$-approximation to the unknown quantum object in a metric that will be clear when used.}

\begin{restatable}{theorem}{learnchannels}\label{theo:learnchannels}
    Let $\Phi$ be a $n$-qubit degree-$d$ quantum channel. There is an algorithm that $(\varepsilon,
    \delta)$-learns~$\Phi$ (in $\ell_2$-distance)
    using 
    $ 
    \exp\big(\widetilde{O}(d^2+d\log (1/\varepsilon))\big)\cdot \log(1/\delta) 
    $ queries  
    to $\Phi$.
\end{restatable}
We remark that the sample complexity of learning general quantum channels requires $\Omega(4^n)$ queries, but if we are promised the channel is \emph{low-degree}, then our  algorithm is much faster than the general algorithm. Additionally  observe that the sample complexity of our learning algorithm is independent of $n$, in contrast to the results~\cite{huang2022learning, volberg2023noncommutative,slote2023bohnenblust,slote2023noncommutative,klein_et_al:LIPIcs.ITCS.2024.69} on quantum learning of low-degree observables, which also are based   on the $\BH$ inequality and have a logarithmic dependence on~$n$.

To sketch the proof the theorem, we first note the  well-known fact that the matrix $\widehat\Phi$ (whose entries are given by $\widehat{\Phi}_{x,y}=\widehat{\Phi}(x,y)$) is a density matrix of a state that is unitarily equivalent to the $\CJ$ state of the channel $\Phi$ \cite[Lemma 8]{bao2023nearly}. Hence, $\widehat\Phi$ can be prepared by querying $\Phi$ just once, and $\{\widehat\Phi(x,x)\}_{x}$ is a probability distribution. The high-level idea behind the learning algorithm is the following.
\begin{enumerate}
    \item Prepare $T$ copies of $\widehat\Phi$ and measure them in the computational basis, allowing the learner to sample from  the distribution $\{\widehat\Phi(x,x)\}_x$.
    \item Using a well-known result in distribution learning theory, we observe that $O(1/\alpha^2)$ many samples from $\{\widehat\Phi(x,x)\}_x$ suffices to obtain the $x$'s such that $\Phi(x,x)$ is $\alpha$-large.
    \item Approximate all the ``large" Pauli coefficients in the step above using a SWAP test for mixed states that allows to estimate the Pauli coefficients $\widehat\Phi(x,y)$. 
    \item Output a $\widetilde{\Phi}$ with coefficients that were estimated above and the remaining coefficients set~to~$0$.
\end{enumerate}
At this point, we use the $\BH$ inequality for quantum channels and show that, as long as $T\sim \exp(d^{2}/\varepsilon^d)$, then the output $\widetilde{\Phi}$ is $\varepsilon$-close to the target $\Phi$ in $\ell_2$-distance.

\paragraph{Remark on learning Pauli channels.} One can also consider a special case of quantum channels called \emph{Pauli channels} $\Phi$, motivated by the fact that Pauli channels are often the dominant noise on quantum devices and a practical noise model for analyzing fault-tolerance~\cite{terhal2015qec}. For these channels, the only non-zero terms in the Fourier expansion of Eq.~\eqref{def:arb_channel} are $\widehat{\Phi}(x,x)$ (i.e., $\widehat{\Phi}(x,y)=0$ when $x\neq y$) and these coefficients are often called \emph{error rates}. Since learning Pauli channels is an important task on near-term quantum devices for error mitigation~\cite{van2023probabilistic} and analyzing error correction, it is desirable to avoid using entangled copies of a state and access to ancillary qubits as part of the learning~algorithm.  With this requirement, it was shown that, in order to $\varepsilon$-learn (in diamond norm) an unknown  $n$-qubit Pauli channels using unentangled measurements, one needs to use the channel  $\Omega(4^n/\varepsilon^2)$ many times~\cite{fawzi2023lower}. In this work, we show that the subclass of low-degree Pauli channels are efficiently learnable. The degree of a Pauli channel $\Phi$ is the minimum integer $d$ such that the error rates $\widehat{\Phi}(x,x)=0$ if $|x|>d$.

Noise on current large-scale quantum devices is modeled as Pauli channels containing a sparse set of local Paulis~\cite{van2023probabilistic}, which is a subclass of low-degree Pauli channels. Such models are available from device physics and experiments used to characterize noise on quantum hardware. When non-local interactions but only over few qubits are included in the Pauli channel~\cite{tran2023locality}, the corresponding Paulis are still low-degree which fits into the class of Pauli channels considered. Our learning result is as follows.

\begin{restatable}{proposition}{learnPauliChannels}\label{theo:learnPauliChannels}
    Let $\Phi$ be an $n$-qubit degree-$d$ Pauli channel. There is an algorithm that $(\varepsilon,
    \delta)$-learns~$\Phi$ (under the diamond norm) using $O\left(n^{2d}/{\varepsilon^2} \cdot \log(n/{\delta}) \right)$ queries to $\Phi$. Moreover, the learning algorithm only requires preparation of  product states and measurements in the Pauli~basis.
\end{restatable}
We remark that the dependence on $n$ in our sample complexity matches  the algorithm in~\cite{flammia2021pauli} for learning low-degree Pauli channels under the diamond norm but our analysis differs from theirs; our result is obtained using a Fourier-analytic approach in contrast to the the result of~\cite{flammia2021pauli} which uses ideas from population recovery. 

\subsection{Result 2: Learning unitaries}
Apart from learning channels, in this paper we also consider the task of learning unknown $n$-qubit unitaries. Similar to the case for channels, it is well-known that learning an unstructured $n$-qubit unitary requires $\widetilde\Theta(4^n)$ applications of $U$~\cite{gutoski2014process}. This complexity can be significantly improved if structural information is available. For example, it has been shown that efficient learning is possible if the unitary corresponds to Clifford circuits~\cite{low2009learning}, corresponds to Clifford circuits with few non-Clifford gates~\cite{lai2022learn}, or are the diagonal unitaries of the Clifford hierarchy~\cite{arunachalam2022optimal}.

In a recent work, Chen et al.~\cite{chen2023testing} considered the task of learning unitaries $U$ that are $k$-juntas and showed that this class can be learned by querying the unitary $\widetilde\Theta(4^k)$ many times.  In this work, we consider the scenario where the unitary is degree-$d$ and show such an exponential saving (in comparison to naive tomography) is possible. This structured class of low-degree unitaries occur in many instances. For example, in nature, the dynamics of many physical systems are governed by local Hamiltonians, whose unitary time evolution operators for short time evolution are close to being low-degree unitaries. Moreover, through the application of Lieb-Robinson bounds to structured many-body Hamiltonians, the corresponding unitary evolution operator can be seen to only have support only on low-degree Paulis~\cite{chen2023speed}. The learning algorithm that we will present is thus applicable for learning and verifying such unitaries.

Consider the Pauli decomposition of an $n$-qubit unitary as follows:
$$
U=\sum_{x\in \{0,1,2,3\}^n} \widehat U(x)\sigma_x,
$$ where $\widehat U(x)$ are its Pauli coefficients. The degree of $U$ is the minimum integer $d$ such that $|x|>d$ implies $\widehat U(x)=0$. Our learning model is similar to the one by Chen et al.~\cite{chen2023testing}. Given an unitary $U$, a learning algorithm is allowed to make \emph{queries to $U$} (and to control-$U$) as follows: it can choose a state $\rho$, apply $U$ to the state to obtain $U\rho U^*$, and measure $U\rho U^*$ in a chosen basis. From the measurement outcomes, the learner should output a classical description of an operator~$\widetilde U$ that is close to $U$ in the $\ell_2$-distance determined by the usual inner product for operators defined as $\langle U, V\rangle=\Tr[U^* V]/2^n$.

\begin{restatable}{theorem}{learnunitaries}\label{theo:learnunitaries}
    Let $U$ be an $n$-qubit degree-$d$  unitary. There is an algorithm that $(\varepsilon,
    \delta)$-learns $U$ (in $\ell_2$-distance) using the unitary $U$  
    $
    \exp\big(\widetilde{O}(d^2+d\log (1/\varepsilon))\big)\cdot \log(1/\delta)
    $  many times.
\end{restatable}
The proof of~\cref{theo:learnunitaries} follows the same structure as the one of~\cref{theo:learnchannels}, but now we learn the Pauli coefficients via the algorithm of Montanaro and Osborne~\cite{montanaro2008quantum}, and the correctness of the algorithm relies on the non-commutative $\BH$ inequality of Volberg and Zhang~\cite{volberg2023noncommutative}. 

The $\BH$ inequality of Volberg and Zhang~\cite{volberg2023noncommutative} works for matrices with bounded operator norm, of which unitaries are a very special case. 
As argued by Montanaro and Osborne~\cite{montanaro2008quantum}, matrices bounded on the operator norm are the quantum analogue of bounded functions $f:\{-1,1\}^n\to [-1,1]$, while unitaries are the  analogue of Boolean functions\footnote{To be precise, they argue that unitary hermitian matrices are the analogue of Boolean functions.} $f:\{-1,1\}^n\to \{-1,1\}$. These two families of classical functions differ a lot with respect to the $\BH$ inequalities: the best upper bound~\cite{defant2019fourier} for the $\BH$ constant for bounded functions is $\exp({d^{1/2}})$, while for Boolean functions of degree $d$ one can even prove that the \emph{usually much bigger} quantity $\sum_{s}|\widehat f(s)|$ is at most $2^{d-1}$~\cite[Exercise 1.11]{o2014analysis}. An open question is if this claim can be generalized to the quantum~setting. 

\begin{question}\label{que:unitariesaresparse}
    Is there a constant $C(d)$ such that $\sum_{x\in \{0,1,2,3\}^n} |\widehat U(x)|\leq C(d)$  for every $n$-qubit degree-$d$ unitary $U$?
\end{question}
If \cref{que:unitariesaresparse} was answered positively, then one could improve the $\eps$ dependence of \cref{theo:learnunitaries} to $(1/\eps)^{2}$. Some evidence in favor of an affirmative answer to \cref{que:unitariesaresparse} is that if Montanaro and Osborne's Conjecture was true~\cite[Conjecture~4]{montanaro2008quantum}, then every degree-$d$ unitary would be a $2^d$-junta, which would imply an affirmative answer with $C(d)=2^{2^d}$.

\paragraph{Remark on learning low-degree quantum states.}
Learning quantum states has been an active line of research given its fundamental importance and applications to quantum system characterization, assessing quality of quantum gates, verification of quantum circuits and validating performance of quantum algorithms. Breakthrough results of Haah et al.~\cite{haah2017sample} and O'Donnell and Wright~\cite{o2016efficient} showed that the sample complexity of learning an unknown $n$-qubit state, up to trace distance $\varepsilon$ is $\Theta(4^n/\varepsilon^2)$. 
A natural consideration is the task of learning \emph{low-degree} quantum states. 
In order to describe this, we first write down the Pauli expansion of an $n$-qubit state $\rho$  as
\begin{equation}\label{eq:PauliExpansion}
    \rho=\sum_{x\in \{0,1,2,3\}^n} \widehat\rho(x) \sigma_x.
\end{equation}
Then, we say that $\rho$ has degree at most $d$ if $\widehat\rho(x)=0$ for all $|x|>d$. 
It is not too hard to see that one can use the formalism of classical shadows~\cite{huang2020predicting} to obtain a learning algorithm (in trace norm)  that has a sample complexity of $\widetilde{O}(n^d/\eps^{2}\cdot \log(n/\delta))$. A similar result (with a different norm and with a bit more structure than just being low-degree) was noted in a recent work of Nadimpalli et al.~\cite{nadimpalli2023pauli}, where they used it to give efficient algorithms to learn QAC$^0$ circuits.

\subsection{Result 3: Learning quantum query algorithms}
Eskenazis and Ivanisvili~\cite{eskenazis2022learning} established a surprising connection between the $\BH$ inequality and learning theory. They considered the following question: suppose $f:\pmset{n}\rightarrow [-1,1]$ is a bounded degree-$d$ function, and a learner is given uniformly random $x$ and $f(x)$, then how many $(x,f(x))$ suffices to learn $f$ up to error $\eps$ in $\ell_2^2$ error? The seminal low-degree algorithm by Linial, Mansour and Nisan uses $O_{d,\eps}(n^d)$ samples \cite{linial1993constant}.\footnote{Here and below, we use $O_{d,\varepsilon}$ to hide factors that depend on  $d,1/\varepsilon$ and independent of $n$.}  This was not improved until recently, when Iyer et al.~\cite{iyer2021tight} reduced this complexity to $O_{d,\eps}(n^{d-1})$. In a surprising work,~\cite{eskenazis2022learning} showed that one can learn $f$ in sample complexity $O_{d,\eps}(\log n)$. For the particular case of bounded $d$-linear tensors $T:(\{-1,1\}^n)^d\to [-1,1]$ they showed that it suffices to use 
\begin{equation}\label{eq:eskenazis2d/d+1}
    (1/\varepsilon)^d\cdot\Big(\sum_{i_1,\ldots,i_d=1}^n |\widehat T_{i_1,\ldots,i_d}|^{2d/d+1}\Big)^{(d+1)/2d}\cdot \log n
\end{equation}
samples $(x,T(x))$, where $x$ is uniform from $(\{-1,1\}^n)^d$ and $\widehat T$ is the tensor of coefficients of $T$, i.e., $T(x)=\sum_{i_1,\dots,i_d}\widehat T_{i_1,\dots,i_d}x_1(i_1)\dots x_d(i_d)$. Combining that with the upper bound of the $\BH$ constant for multilinear tensors~\cite{bayart2014bohr}, it yields
\begin{equation}\label{eq:eskenazisbounded}
    (d/\eps)^{O(d)}\cdot O(\log n)
\end{equation}
uniformly random samples are enough to learn $T$. Although this result is surprising since the complexity only scales polylogarithmic with $n$, observe that if $d=\omega(\log n)$, then the sample complexity is superpolynomial in $n$. This motivates the following question: are there classes of polynomials that can be learned using $\poly(n)$ samples for any $d=\omega(\log n)$? Below we answer this question positively for the class of polynomials that arise from quantum query algorithms.

\paragraph{Quantum polynomials.}
The result of \cref{eq:eskenazisbounded} can be applied to learn the amplitudes of quantum algorithms that query different blocks of variables every time (see \cref{fig:QQA}), as they are multilinear tensors bounded on the supremum norm~\cite{polynomialmethod}. To be precise, we consider quantum query algorithms such that they prepare a state $$\ket{\psi_x}=U_d(O_{x_d}\otimes \Id_m) U_{d-1}\dots U_1(O_{x_1}\otimes \mathbb \Id_m)U_0 \ket u,$$
where $m$ is an integer,  $x$ stands for $(x_1,\dots,x_d)$, $O_y$ is the $n$-dimensional unitary that maps $\ket i$ to $y_i\ket i$, $U_0,\dots,U_d$ are $(n+m)$-dimensional unitaries and $\ket u$ is a $n+m$-dimensional unit vector. The algorithm succeeds according to a projective measurement that measures
the projection of the final state onto some fixed $n+m$ dimensional unit vector $\ket v$. Hence, the amplitude of $\ket v$ is $T(x)=\langle v|\Psi_x\rangle$, so $|T(x)|^2$ is the acceptance probability of the algorithm.
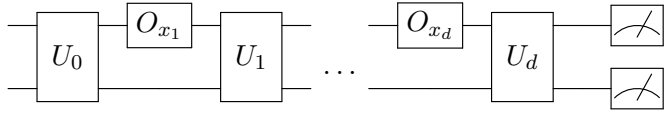
\begin{figure}[!ht]
    \centering
    \hspace{0.5cm}
	\Qcircuit @C=1em @R=.7em {
		\ \ \ \ & \multigate{1}{U_0} & \gate{O_{x_1}} & \multigate{1}{U_1}&\qw &  &  & \gate{O_{x_d}} & \multigate{1}{U_d} &\qw & \meter\\
		\ \ \ \ & \ghost{U_0} & \qw & \ghost{U_1} &\qw &\ustick{\dots} & &   \qw & \ghost{U_d} & \qw &\meter \\
	}
	\caption{Quantum query algorithms considered in \cref{theo:learnqpolynomials}.}\label{fig:QQA}
\end{figure}
These quantum algorithms have been considered in the quantum computing literature. For example, $k$-forrelation, that witnesses the biggest possible quantum-classical separation, has this structure~\cite{aaronson2015forrelation,bansal2021k}. Also, for these algorithms the Aaronson and Ambainis conjecture is known to be true, so they can be classically and efficiently simulated almost everywhere~\cite{Bansal:2022,gutierrez2023influences}.

In addition,  it was shown~\cite{QQA=CBF} that those amplitudes are not only bounded, but also completely bounded. 
Our main contribution regarding these algorithms is showing that for $d$-linear tensors $T$ that are \emph{completely bounded}, we can improve the $\BH$ inequality~to
$$
\Big(\sum_{i_1,\ldots,i_d=1}^n |\widehat T_{i_1,\ldots,i_d}|^{2d/d+1}\Big)^{(d+1)/2d}\leq 1.
$$ 
Using this upper bound, we show the following.
\begin{restatable}{theorem}{learnqpolynomials}\label{theo:learnqpolynomials}
    For a quantum algorithm that makes $d$-queries as in \cref{fig:QQA},  there is a classical algorithm that $(\eps,\delta)$-learns its amplitudes  (in $\ell_2^2$-distance) using  
    $
    O\left((1/\varepsilon)^d\cdot  \log n\right)
    $
    uniform samples.
\end{restatable}
This exponentially improves the complexity of~\cite{eskenazis2022learning}, as stated in Eq.~\eqref{eq:eskenazisbounded} for the natural class of  polynomials arising from quantum query algorithms. In particular, for $d=\omega(\log n)$ and constant~$\eps$, one can learn this class of polynomials with sample complexity that is polynomial in $n$.

\subsection{Result 4: Quantum learning of classical polynomials}
\paragraph{Boolean functions.} Quantum learning not only concerns quantum objects, but also classical ones. For instance, Boolean functions
 $f:\{-1,1\}^n\to \{-1,1\}$ can be accessed using a \emph{quantum uniform example}, given by 
 \begin{equation}\label{eq:quantumsamples}
     \frac{1}{\sqrt{2^n}}\sum_{x\in \{-1,1\}^n}\ket{x,f(x)}.
 \end{equation} 
This data access model has been vastly studied in the literature, where several notable quantum speedups have been proven~\cite{bshouty1995learning,atici2007quantum,arunachalam2021two}. Many of these speedups are analyzed trough the Fourier transform, that allows to identify every function $f:\{-1,1\}^n\to\mathbb R$ with a multilinear polynomial  $f=\sum_{s\in \{0,1\}^n}\widehat f(s)\chi_s$, 
where $\widehat f(s)\in \mathbb R$ are the Fourier coefficients and $\chi_s$ are the character functions $\chi_s(x)=\prod_{i\in [n]}x_i^{s_i}$. The degree for these functions is the minimum integer $d$ such that if $|s|>d$ then $\widehat f(s)=0$. It is a well-known fact that Boolean functions $f:\{-1,1\}^n\to \{-1,1\}$ of degree $d$ are $2^{1-d}$-granular, meaning that their Fourier coefficients lie in $2^{1-d}\mathbb Z$. This has immediate consequences for both learning theory and $\BH$ inequalities. 
\begin{restatable}{fact}{learnBool}\label{fact:learnBool}
   Let $f:\{-1,1\}^n\to\{-1,1\}$ be a degree-$d$ function. There is a quantum algorithm that learns $f$ exactly  with probability 1$-\delta$ using $O\left(4^dd\log\left(1/\delta\right)\right)$ quantum uniform examples. In contrast, any classical algorithm requires $\Omega(2^d+\log n)$ uniform examples and there is a classical algorithm that only uses $O\left(4^dd\log\left(n/\delta\right)\right)$ uniform examples.
\end{restatable}
Despite the simplicity of the proof of \cref{fact:learnBool}, we state it for completeness and because it seems not to be well-known (see for instance~\cite[Corollary 34]{nadimpalli2023pauli}, which proposes a quantum algorithm for the same problem that requires $O(n^d)$ samples, or \cite[Corollary 4]{eskenazis2022low} that proposes a classical algorthm that requires $O(2^{d^2}\log n)$ samples). We also remark that the classical lower bound was recently improved from $\Omega(2^d+\log n)$ (which is folklore, and it is a lower bound for learning Boolean $d$-juntas) to $\Omega(2^d\log n)$ \cite{eskenazis2022low}.  Here, we observe a $\BH$-type inequality for Boolean functions.

\begin{restatable}{fact}{BHBool}\label{fact:BHBool}
    Let $f:\{-1,1\}^n\to\{-1,1\}$ of degree at most $d$. Then, 
    $$
    \Big(\sum_{s\in\{0,1\}^n}|\widehat f(s)|^{\frac{2d}{d+1}}\Big)^{\frac{d+1}{2d}}\leq  2^{\frac{d-1}{d}}.$$
    The equality is witnessed by the address function.     
\end{restatable}
\cref{fact:BHBool} might be of interest in functional analysis for two reasons: $(i)$ it is conjectured that the value of the $\BH$ constant for $d$-linear tensors is $2^{\frac{d-1}{d}}$~\cite{pellegrino2018towards} (tensors correspond to multilinear forms in the functional analysis language), so this fact proves the conjecture for particular case of $d$-linear Boolean tensors, 
$(ii)$ the address function\footnote{ For $d\in\mathbb N$, the address function $f:(\{-1,1\}^2)^{d-1}\times\{-1,1\}^{2^{d-1}}\to \{-1,1\}$ is defined as 
    $f(x,y)=\sum_{a\in \{-1,1\}^{d-1}}g_a(x)y(a),$
    where we identify $\{-1,1\}^{d-1}$ with $[2^{d-1}]$ and $g_a(x)$ is $0$ unless $x_i(1)=-a_ix_i(2)$ for every $i\in [n]$, in which case it takes the value $\prod_{i\in [n]} x_i(1)$ .}, that saturates the inequality, is a $d$-linear form that gives a lower bound for the $\BH$ constant for multilinear tensors of $2^{\frac{d}{d-1}},$ which matches the best lower bound known so far~\cite{diniz2014lower}.

\paragraph{Boolean polynomials.} For bounded functions $f:\{-1,1\}^n\to [-1,1]$ the definition of quantum uniform examples defined as in \cref{eq:quantumsamples} is unclear. Given that bounded polynomials have received attention in a few works~\cite{linial1993constant,iyer2021tight,eskenazis2022learning}, we propose a way learning them quantumly by accessing these polynomials through a block encoding~\cite{chakraborty2018power}, i.e., a learning algorithm has access to a block encoding of the $2^n$-dimensional diagonal matrix whose diagonal entries all equal~$f$. To this end, we prove the following theorem. 

\begin{restatable}{proposition}{learnBlockEncoding}\label{prop:learnBlockEncoding}
    Let $f:\{-1,1\}^n\to [-1,1]$ be a degree-$d$ polynomial. There is an algorithm~that $(\varepsilon,
    \delta)$-learns $f$ (in $\ell_2$-distance)  using $\exp(\widetilde{O}(\sqrt{d^{3}}+d\log(1/\eps)) \log (1/\delta))$ copies of a block-encoding~of~$f$.
\end{restatable}
The proof of \cref{prop:learnBlockEncoding} is a combination of the ideas in~\cite{eskenazis2022learning} with the Fourier sampling and block-encoding quantum primitives. We remind the reader that the merit of Eskenazis and Ivanisvili was to bring down the classical complexity of the problem from $O_{d,\eps}(n^d)$ to $O_{d,\eps}(\log n)$. \cref{prop:learnBlockEncoding} shows that the quantum complexity could be even reduced to $O_{d,\eps}(1)$. \cref{prop:learnBlockEncoding} also implies a quantum speedup (with respect to $n$), as the lower bound of $\Omega(2^d+\log n)$ also holds for membership queries\footnote{Making a membership query to $f$ consists on accessing one pair $(x,f(x))$ where $x$ is chosen by the learner, not necessarily uniformly at random. If $f$ is a Boolean function and $U_f$ is the unitary defined by $U_f\ket{x}=f(x)\ket{x}$, then a membership query $(x,f(x))$ can be simulated by applying $(H\otimes \Id_n)CU_f(H\otimes \Id_n)$ to $\ket{0}\ket{x}$ and measuring the first qubit in the computational basis. Note that $U_f$ can be regarded as a block-encoding of $f$.}, which are the classical analogue of accessing a unitary block-encoding of $f$. 
Additionally, we highlight that although our results are about sample complexity,  the time complexity of our quantum algorithm of \cref{prop:learnBlockEncoding} scales as $\poly(n,\exp(d^{1.5}))$, while the current state-of-the-art classical algorithm~\cite{eskenazis2022learning} approximates the ${n \choose d}$ Fourier coefficients, so it requires $\poly(n^d)$ time. In particular, for $d=(\log n)^{2/3}$, constant $\eps,\delta$, our quantum algorithm has a time complexity of $\poly(n)$ time, while the state-of-the-art classical algorithm runs in time $n^{\mathrm{polylog} n}$. 

\paragraph{Future directions.}
In this work we have considered the normalized $\ell_2$-norm as a  measure of accuracy of the learning algorithms. This norm is standard in learning theory if $f$ and $g$ are Boolean functions, since  $\mathrm{P}_{x}[g(x)\neq h(x)]\leq \norm{f-g}_2^2,$ so learning under the $\ell_2$-distance implies $\textsf{PAC}$ learning under the uniform distribution. In the quantum case, the normalized $\ell_2$-norm can be interpreted as average-case distances and the $\ell_1$-norm represent worst-case distances and usually has a better operational interpretation. Thus, in the quantum realm it would desirable to obtain $n$-independent learning results with respect to $\ell_1$-distances. Our work leaves this as an open question.

\paragraph{Acknowledgements.}
We thank the referees of ICALP'24 and TQC'24 for useful comments. We thank Jop Bri\"et and Antonio Pérez Hernández for useful discussions. We thank Alexandros Eskenazis, Paata Ivanisvili, Alexander Volberg and Haonan Zhang for useful comments. A.D. and S.A. thank the Institute for Pure and Applied Mathematics (IPAM) for its hospitality throughout the long program “Mathematical and Computational Challenges in Quantum Computing” in Fall 2023 during which part of this work was initiated.
This research was supported by the European Union’s Horizon 2020 research and innovation programme under the Marie Sk{\l}odowska-Curie grant agreement no. 945045, and by the NWO Gravitation project NETWORKS under grant no. 024.002.003. C.P. is partially supported by the MICINN project PID2020-113523GB-I00, by QUITEMAD+-CM, P2018/TCS4342, funded by Comunidad de Madrid and by Grant CEX2019-000904-S funded by MCINN/AEI/ 10.13039/501100011033.

\section{Preliminaries}\label{sec:prelims}
Throughout this work we will consider different constants, all of which will be denoted by $C$ and their value will be clear from context. We will use $C(d)$ to refer to quantities that only depend $d$ and are constant with respect to other parameters. Given a natural number $m$, $[m]$ represents the set $\{1,2,\dots,m\}$ and $[m]_0$ represents $\{0\}\cup [m].$

We will work with $n$-qubit systems throughout this work and denote the dimension of the underlying Hilbert space as $N=2^n$. We denote the $n$-qubit Pauli group as $\mathcal{P}_n$ which is the set of $n$-fold tensor product of single qubit Pauli operators $\{\pm 1, \pm i\} \times \{I,X,Y,Z\}^{\mathop{\otimes} n}$. We will denote Pauli operators (up to a phase) by $\sigma_x=\otimes_{i\in [n]}\sigma_{x_i}$, where $x \in \{0,1,2,3\}^n$ and the single qubit Pauli operators are defined by $\sigma_0 = I$, $\sigma_1 = X$, $\sigma_2 = Y$, and $\sigma_3=Z$. For example, the $4$-qubit Pauli operator $XYZI$ is equivalently denoted by $\sigma_s$ with $s=(1,2,3,0)$. This convention to parameterize the Pauli operators in $\calP_n$ (up to a phase) by strings in $\{0,1,2,3\}^n$ will be useful in defining the Pauli expansion of quantum states and unitaries later. For any $n$-qubit Pauli operator $\sigma_x$, we call the size of the support of the Pauli over non-identity single qubit Paulis as its \textit{degree}, denoted by $|x|$. A particular set of states that we will use throughout this work are the eigenstates of the Pauli operators. We denote the $\pm 1$ eigenstates of the single-qubit Pauli operators $\sigma_b$ as $\ket{\chi_{\pm 1}^b}$, where $b \in \{1,2,3\}$. Product states over eigenstates can then be denoted by $\rho_{s,a} = \mathop{\otimes}_{j=1}^n \ket{\chi_{a_j}^{s_j}}\bra{\chi_{a_j}^{s_j}}$ with $s \in \{1,2,3\}^n$ and $a \in \{+1,-1\}^n$ i.e., they correspond to strings $s$ and~$a$.

We will write $\mathcal{M}_{N}$ to denote the set of linear operators from $\mathbb{C}^N$ to $\mathbb{C}^N$. The set of $n$-qubit quantum unitary operators can then be defined as
\begin{equation}
    \mathcal{U}_N := \{U \in \mathcal{M}_{N} : UU^* = U^* U = I\}.
\end{equation}
We define $n$-qubit to $n$-qubit quantum channels $\Phi$ to be completely positive and trace-preserving maps from $\mathcal{M}_{N}$ to $
\mathcal{M}_{N}$.

\paragraph{Choi-Jamiolkowski isomorphism.}
We will often use the Choi-Jamiolkowski isomorphism to encode a quantum unitary or quantum channel as a quantum state. We call the resulting state as the Choi-Jamiolkowski state (or $\CJ$ state for short). 
In the case of unitaries $U$, the $\CJ$ state will be denoted by $\ket{v(U)}$, which is
\begin{equation}
    \ket{v(U)} := (U \mathop{\otimes} I) \left(\frac{1}{\sqrt{N}} \sum_{i \in [N]} \ket{i} \ket{i} \right) = \frac{1}{\sqrt{N}} \sum_{i,j \in [N]} U[i,j] \ket{i} \ket{j}.
    \label{eq:CJ_state}
\end{equation}
The CJ state $\ket{v(U)}$ can be prepared by first preparing $n$ EPR pairs (over $2n$ qubits) and the applying the unitary $U$ to the $n$ qubits coming from the first half of each of the $n$ EPR pairs. Noting the Fourier decomposition of the unitary $U = \sum_{x \in \{0,1,2,3\}^n} \widehat{U}(x) \sigma_x$ (which will be presented shortly), we can also equivalently express the CJ state of $U$ in terms of the $\CJ$ states corresponding to the different Pauli operators $\ket{v(\sigma_x)}$ as follows:
\begin{equation}\label{eq:CJBellExpansion}
    \ket{v(U)} = \sum_{x \in \{0,1,2,3\}^n} \widehat{U}(x) \ket{v(\sigma_x)}.
\end{equation} 
In case of quantum channels $\Phi$, the $\CJ$ representation is given by
\begin{equation}
    J(\Phi) = \sum_{i,j \in [N]} \Phi \left(\ket{i}\bra{j}\right) \mathop{\otimes} \ket{i}\bra{j} = (\Phi \mathop{\otimes} I) \left(\sum_{i,j \in [N]} \ket{i}\bra{j} \mathop{\otimes} \ket{i}\bra{j} \right),
\end{equation}
with the CJ state $v(\Phi)$ then defined to be
\begin{equation}
    v(\Phi) = \frac{J(\Phi)}{\Tr[J(\Phi)]}=\frac{J(\Phi)}{N}.
    \label{eq:choi_state_channel}
\end{equation}

\subsection{Fourier and Pauli analysis}
In this section, we describe Fourier expansion of Boolean functions and of different quantum objects (states, unitaries, channels) that we consider throughout this work. Note that the terms \emph{Pauli expansion} and \emph{Fourier expansion} will often be used interchangeably in the context of quantum objects .

\paragraph{Fourier expansion.}
In this section we will talk about the space of  functions defined on the Boolean hypercube $f:\{-1,1\}^n\to \mathbb R$ endowed with the inner product $\langle f,g\rangle=\mathbb E_x [f(x)g(x)]$, where the expectation is taken with respect to the uniform measure of probability. For $S\subseteq [n]$, the Fourier characters, defined by $\chi_S(x)=\prod_{i\in  S}x_i$, constitute an orthonormal basis of this space. Hence, every $f$ can be identified with a multilinear polynomial via the Fourier expansion
\begin{equation}\label{eq:FourierExpansion}
    f=\sum_{s\in \{0,1\}^n}\widehat f(s)\chi_s,
\end{equation}
where $\widehat f(s)$ are the Fourier coefficients given by
\begin{equation}\label{eq:FourierExpansion2}
    \widehat f(s)=\langle \chi_s,f\rangle=\mathbb E_{x} [f(x)\chi_s(x)].
\end{equation}
The degree of $f$ is the minimum $d$ such that $\widehat f(s)=0$ if $|s|>d$. We will often use  Parseval's~identity: 
\begin{equation}\label{eq:ParsevalIdentity}
    \norm{f}_2^2:=\langle f,f\rangle =\sum_{s\in [n]}\widehat f(s)^2.
\end{equation}
We will also consider the $p$-norms of the Fourier spectrum, which are defined as $$\norm{\widehat f}_p=\left(\sum_{s\in\{0,1\}^n}|\widehat f(s)|^p\right)^{1/p}.$$

\paragraph{Pauli expansion of operators.} In this section we consider $\mathcal{M}_N$ endowed with the usual inner product $\langle A,B\rangle =
\frac{1}{N}\Tr[A^* B]$, where here  and in the following $A^*$ denotes the adjoint matrix of $A$.  The Pauli operators $\sigma_x$
form an orthonormal basis for this space.
The Pauli expansion of a matrix~$M$ of $\mathcal{M}_N$  is given by
\begin{equation}
    M = \sum_{x \in \{0,1,2,3\}^n} \widehat M(x) \sigma_x, 
\end{equation}
where $\widehat{M}(x) = \langle \sigma_x,M\rangle$.
are Pauli coefficients of $M$. We will refer to the collection of non-zero Pauli coefficients $\{\widehat{M}(x)\}_x$ as the Pauli spectrum of $M$ with the set of corresponding strings denoted by $\mathrm{spec}(M)$. As $\{\sigma_x\}_x$ is an orthonormal basis, we have a version of Parseval's identity for operators. 
\begin{equation}
    \norm{M}_2^2:=\langle M,M\rangle=\sum_{x\in \{0,1,2,3\}^n}|\widehat M(x)|^2.
\end{equation}
In particular, for $U\in\mathcal{U}_N,$ this implies that  $(|\widehat U(x)|^2)_x$ is a probability distribution. 
We will also consider the $p$-norms of the Pauli spectrum, which are defined as $$\norm{\widehat M}_p=\left(\sum_{x\in\{0,1,2,3\}^n}|\widehat M(x)|^p\right)^{1/p}.$$
We now define a notion of degree for states and unitaries that generalizes the classical notion of Fourier degree (see~\cite[Section 5]{montanaro2008quantum}).

\begin{dfn}[Degree of a matrix]\label{def:deg_rho}
Given $M\in\mathcal{M}_N$ its degree is the minimum $d$ such that $\widehat{M}(x)=0$ for any $x \in \{0,1,2,3\}^n$ with $|x|>d$.
\end{dfn} 
We recall that this analogy between between Fourier analysis and Pauli analysis for operators was first explored by Montanaro and Osborne \cite{montanaro2008quantum}.
\paragraph{Pauli expansion of superoperators.} We consider the space of superoprators (linear maps from $\mathcal{M}_N$ to $\mathcal{M}_N$) endowed with the inner product $\langle \Phi, \Psi \rangle = \langle J(\Phi), J(\Psi)\rangle/N^2$.  An orthonormal basis for superoperators is defined using characters
\begin{equation}
    \Phi_{x,y}(\rho) = \sigma_x \rho \sigma_y,
\end{equation}
for any $x,y \in \{0,1,2,3\}^n$. The Pauli expansion of superoperators and hence quantum channels is then defined as
\begin{equation}
    \Phi = \sum \limits_{x,y \in \{0,1,2,3\}^n} \widehat{\Phi}(x,y) \Phi_{x,y},
\end{equation}
where  $\widehat{\Phi}(x,y) = \langle \Phi_{x,y}, \Phi \rangle$ are the Pauli coefficients of the superoperator. As $\{\Phi_{x,y}\}_x$ is an orthonormal basis, we have a version of Parseval's identity for superoperators
\begin{equation*}
    \norm{\Phi}_2^2:=\langle \Phi,\Phi\rangle=\sum_{x,y\in\{0,1,2,3\}^n}|\widehat\Phi(x,y)|^2.
\end{equation*} 
We will also consider the $p$-norms of the Pauli spectrum of superopertors, which are defined as $$\norm{\widehat M}_p=\left(\sum_{x,y\in\{0,1,2,3\}^n}|\widehat \Phi(x,y)|^p\right)^{1/p}.$$
If $\Phi$ is a channel, then $\widehat\Phi=(\widehat\Phi(x,y))_{x,y}$ has a couple of important properties \cite[Lemma 8]{bao2023nearly}. 
\begin{fact}
    If $\Phi$ is a channel, then $\widehat\Phi$ is a state unitarily equivalent to $v(\Phi)$. In particular, $(\widehat\Phi(x,x))_x$ is a probability distribution. 
\end{fact}
\noindent The degree of a superoprator is defined in the analogue way to operators. 
\begin{dfn}[Degree of a superoperator]\label{def:superop}
Given a superoperator $\Phi$ its degree is the minimum $d$ such that $\widehat{\Phi}(x,y)=0$ for any $x,y \in \{0,1,2,3\}^n$ with $|x|+|y|>d$.
\end{dfn}  

We recall that this analogy between Fourier analysis and Pauli analysis for superoperators was first explored by Bao and Yao \cite{bao2023nearly}.

\subsection{Bohnenblust-Hille  inequalities}\label{subsec:introBH} 
\paragraph{Littlewood and Bohnenblust-Hille inequalities.} In 1930 Littlewood~\cite{littlewood1930bounded} showed that for every bilinear form $B:\mathbb C^n\times\mathbb C^n \to \mathbb C$ it holds that 
\begin{equation}\label{eq:littlewood}
    \norm{\widehat B}_{4/3}\leq \sqrt{2}\norm{B}_{\infty},
\end{equation}
where $\widehat B\in \mathbb C^{n\times n}$ are the coefficients of $B$, i.e., $B(x,y)=\sum_{i,j\in [n]}\widehat B_{i,j}x_iy_j$,  and
$$
\|\widehat B \|_{4/3}=\left(\sum_{i,j\in [n]}|\widehat B_{i,j}|^{4/3}\right)^{3/4}\quad \mathrm{and}\quad \norm{B}_\infty=\sup_{|x_i|,
|y_j|\leq 1}\left|\sum_{i,j\in [n]}\widehat B_{i,j} x_iy_j\right|.
$$  
A year later, Bohnenblust and Hille~\cite{bohnenblust1931absolute} generalized the result for $d$-linear tensors $T:(\mathbb C^n)^d\to \mathbb C$ by showing that  there exists a constant $C(d)$,  depending only on $d$ (and not $n$), such that 
\begin{equation}\label{eq:BH}
    \norm{\widehat T}_{\frac{2d}{d+1}}\leq C(d)\norm{T}_{\infty},
\end{equation}
where $\widehat T$ are the coefficients of $T$ and $\|\widehat T\|_{2d/(d+1)}$ and $\norm T_\infty$ are defined in analogue way to the bilinear case.
 These inequalities also hold (possibly with different constants) if the tensors are defined over $\mathbb R$ instead of $\mathbb C$. Determining the optimal constants has received a lot of attention in mathematics~\cite{pellegrino2012new,munoz2012geometric,nunez2013there,pellegrino2018towards,vieira2018optimal,cavalcante2020geometry}.    In terms of upper bounds, it is known that $C(d)
\leq \poly (d)$~\cite{bayart2014bohr}, but the best known lower bound is no bigger than a constant: $1$ for the complex case and $2^{1-\frac{1}{d}}$ for the real case~\cite{diniz2014lower}. It is also known that $p(d)=2d/(d+1)$ is the smallest number such that the inequality $\|\widehat T\|_{p(d)}\leq  C(d)\norm{T}_{\infty}$ holds for some constant $C(d)$ independent of $n$. 

\paragraph{Grothendieck inequality.} 
In 1953, Grothendieck introduced his famous inequality for bilinear maps~\cite{grothendieck1953resume}. It establishes\footnote{We note that the original formulation of the inequality was not the one presented here. In particular, Grothendieck did not mention the cb-norm used in this work.} that there is a constant $K$ such that for every bilinear map $B:\mathbb K^n\times\mathbb K^n\to \mathbb K$, with $\mathbb K\in \{\mathbb R,\mathbb C\}$, it holds that 
\begin{equation}\label{eq:GT}
    \norm{B}_{\cb}\leq K \norm{B}_{\infty}.
\end{equation}
Here $\norm{B}_\cb$ stands for the completely bounded norm and is given by
$$
\norm{B}_{\cb}=\sup\norm{\sum_{i,j\in [n]}\widehat B_{i,j} X(i)Y(j)}_{\op},
$$
where $\norm{\cdot}_{\op}$ stands for the operator norm of matrices (largest singular value), the supremum runs over all integers $m\in\mathbb N$ and all matrices $X(i),\ Y(j)\in \mathcal{B}(\ell_2^m(\mathbb K),\ell_2^m(\mathbb K))$ whose operator norm is at most $1$. We denote the smallest constants that can appear in \eqref{eq:GT} by $\GK$. Determining the exact value of these constants has also been a separate line of research (see~\cite[Section 4]{pisier2012grothendieck} and the references therein). The current state of the art are the bounds $1.676 <\GR<1.782$~\cite{Davie:1984, Reeds:1991, Braverman:2013} and $1.338<\GC<1.4049$~\cite{haagerup1987new,davie2006matrix}. Grothendieck inequality has found a lot of applications in computer science and quantum information, for instance in non-local games \cite{Tsirelson}, quantum-query complexity \cite{Aaronson2015PolynomialsQQ} and matrix approximation theory \cite{alon2004approximating}.

For $d$-linear tensors the completely bounded norm is defined in analogue way: 
\begin{equation}\label{eq:Tcb}
    \norm{T}_{\cb}=\sup\norm{\sum_{\ind i\in [n]^d}\widehat T_{\ind i} X_1(i_1)\dots X_t(i_t)}_{\op},
\end{equation}
where again $X_s(i_s)$ are bounded on operator norm. This norm has been vastly studied in the complex case, where it corresponds with the Haagerup tensor norm of $\ell_1^n(\mathbb C)$ (see for instance~\cite[Chapter 17]{paulsenoperatoralgebras}). In contrast to the Littlewood inequality \eqref{eq:littlewood}, which generalizes to the $\BH$ inequality, one can show that Grothendieck inequality fails for $d$-linear tensors (this follows, for instance, from the stronger results proved in~\cite{perez2008unbounded}). More precisely, there is no constant $K'$ independent of $n$ such that for every trilinear form $T$ it holds $$\norm{T}_{\cb}\leq K'\norm{T}_{\infty}.$$

\paragraph{The completely bounded $\BH$ inequality.} From the $\BH$ inequality and the fact that $\norm{T}_{\infty}\leq \norm{T}_{\cb}$ it follows that there is a constant $C(d)$ such that 
\begin{equation}\label{BHG:forms}
    \norm{\widehat T}_{\frac{2d}{d+1}}\leq C(d)\norm{T}_{\cb}.
\end{equation}
We call this inequality the \emph{completely bounded $\BH$ inequality}. 

\paragraph{Non-commutative $\BH$ inequalities.} Since the breakthrough of Eskenazis and Ivanisvili, the $\BH$ inequality got the attention of the quantum learning theory community. That lead to the non-commutative versions of the $\BH$ inequality~\cite{huang2022learning,volberg2023noncommutative,slote2023bohnenblust,slote2023noncommutative,klein_et_al:LIPIcs.ITCS.2024.69}.  Volberg and Zhang \cite{volberg2023noncommutative} proved that for a matrix $M\in\mathcal{M}_N$ there exists a constant $C$ such~that 
\begin{equation}\label{eq:BHUnitaries}
    \norm{\widehat M}_{\frac{2d}{d+1}}\leq C^{d} \norm{M}_{\op}.
\end{equation}

\section{Learning low-degree quantum channels}
In this section we give our learning algorithms for quantum channels, i.e., we prove \cref{theo:learnchannels,theo:learnPauliChannels}. Before that, we need to prove a new $\BH$ inequality for quantum channels.

\subsection{Bohnenblust-Hille inequality for quantum channels}
In this section we prove a Bohnenblust-Hille inequality for quantum channels.  In fact, it is a result for superoperators that are bounded in the $S_1$ to $S_\infty$ norm (defined below), of which quantum channels are an example. Hence, in this section $\Phi$ will be treated a linear map from $\mathcal{M}_N$ to $\mathcal{M}_N$.  The $S_1$ to $S_\infty$ norm of superoperator is defined by $$\norm{\Phi}_{S_1\to S_\infty}=\sup_{M\neq 0}\frac{\norm{\Phi(M)}_{S_\infty}}{\norm{M}_{S_1}},$$ where $\norm{M}_{S_1}$ is the Schatten $1$-norm of $M$, i.e., the sum of the singular values of $M$, and $\norm{\Phi(M)}_{\infty}$ is the Schatten $\infty$-norm, which coincides with the usual operator norm given by the largest singular value.

In order to prove our theorem will reduce to the classical case of functions $f:\{-1,1\}^n\to \mathbb R$.
\begin{theorem}[\cite{defant2019fourier}]
\label{theo:DMP}
    Let $p:\{-1,1\}^n\to \mathbb R$ of degree at most $d$. Then, $$
    \norm{\widehat p}_{\frac{2d}{d+1}}\leq C^{\sqrt{d\log d}}\norm{p}_\infty,
    $$
    where $C>0$ is a constant.
\end{theorem}

\noindent We follow a similar argument to the one used in  \cite{volberg2023noncommutative}. However, we need to modify their argument in order to considered maps from $S_1$ to $S_\infty$ and not just matrices in $M_N$. Here, for every superoperator $\Phi:\mathcal{M}_M\to\mathcal{M}_N$, we will assign it a function $f_{\Phi}:\{-1,1\}^{3n}\times \{-1,1\}^{3n}\to\mathbb C$ defined as follows. For $a=(a^1,a^2,a^3),\ b=(b^1,b^2,b^3)\in\{-1,1\}^n\times\{-1,1\}^n\times\{-1,1\}^n$ and $s,t\in \{1,2,3\}^n$, define the following matrices (which are not necessarily states)
$$
\ket{a^s}\langle b^{t}|=\mathop{\otimes}_{i\in [n]}\ket{\chi^{s(i)}_{a^{s(i)}_i}}\bra{\chi^{t(i)}_{b^{{t}(i)}_i}},
$$
The function $f_{\Phi}:\{-1,1\}^{3n}\times\{-1,1\}^{3n}\to \mathbb C$ is then given by 
$$
f_{\Phi}( a,b)=\frac{1}{9^n}\sum_{s,t\in \{1,2,3\}^n}\Tr[\Phi\left(\ket{a^s}\langle b^{t}|\right)|b^{t}\rangle\langle a^{s}|],
$$ where $f_{\Phi}$ has the following properties, which allow for a reduction to the classical $\BH$ inequality.

\begin{lemma}\label{prop:propertiesfPhi}
    Let $\Phi$ be a degree-$d$ superoperator. Then,
    $|f_{\Phi}(a,b)|\leq \norm{\Phi}_{S_1\to S_\infty}$ for all $a,b$ and $\norm{\widehat\Phi}_{p}\leq 3^{d}\norm{\widehat f_{\Phi}}_p$. The degree of $f_{\Phi}$ as a multilinear polynomial is $d$. 
\end{lemma}
\begin{proof}
    We first show the bound on $|f_{\Phi}|$. Given that $\ket{a^s}\langle b^{t}|$ is a rank one operator such that $\|\ket{a^s}\|_2=\|\ket{b^t}\|_2=1$, we conclude that
    \begin{equation}\label{eq:S1Sinftynorms}
    \norm{\ket{a^s}\langle b^{t}|}_{S_1}=1.
    \end{equation}
    Thus, we have that $f_\Phi$ is bounded:
    \begin{align*}
        |f_{\Phi}( a,b)|&\leq \frac{1}{9^n}\sum_{s,t\in \{1,2,3\}^n}|\Tr[\Phi\left(\ket{a^s}\langle b^{t}|\right)|b^{t}\rangle \bra{a^{s}}]|\\
        &\leq \frac{1}{9^n}\sum_{s,t\in \{1,2,3\}^n}\norm{\Phi\left(\ket{a^s}\langle b^{t}|\right)}_{S_\infty}\norm{|b^{t}\rangle\bra{a^{s}}}_{S_{1}}\\
        &\leq \frac{1}{9^n}\sum_{s,t\in \{1,2,3\}^n}\norm{\Phi}_{S_1\to S_\infty}\norm{\ket{a^s}\langle b^{t}|}_{S_1}\norm{|b^{t}\rangle\bra{a^{s}}}_{S_{1}}\\
        &\leq \frac{1}{9^n}\sum_{s,t\in \{1,2,3\}^n}\norm{\Phi}_{S_1\to S_\infty}=\norm{\Phi}_{S_1\to S_\infty},
    \end{align*}
    where in the first inequality we have used the triangle inequality, in the second inequality  the duality between $S_1$ and $S_\infty$, in the third the definition of $S_1\to S_\infty$ norm and in the fourth \cref{eq:S1Sinftynorms}.  We now prove that $\norm{\widehat\Phi}_{p}\leq 3^{-d}\norm{\widehat f_{\Phi}}_p$ and that the degree of $f_{\Phi}$ is $d$. For that it suffices to show that 
    \begin{equation}\label{eq:fPhicoefficients}
        f_{\Phi}(a,b)=\sum_{x,y\in \{0,1,2,3\}^n} \frac{\widehat \Phi(x,y)}{3^{|x|+|y|}} \prod_{i\in \mathrm{supp}(x)}\prod_{j\in \mathrm{supp}(y)}a^{x(i)}_ib^{y(j)}_j,
    \end{equation}
    where $\mathrm{supp}(x)=\{i\in [n]:\, x_i\neq 0\}$ and $|x|$ is the size of supp$(x)$. Indeed, this follows from the fact that for every $x,y\in \{0,1,2,3\}^n$, the product $\prod_{i\in \mathrm{supp}(x)}\prod_{j\in \mathrm{supp}(y)}a^{x(i)}_ib^{y(j)}_j$ can be read as $\chi_{S_{x,y}}(a,b)$ for a certain $S_{x,y}\in \{-1,1\}^{6n}$ satisfying that $S_{x,y}\neq S_{x',y'}$ whenever $(x,y)\neq (x',y')$.
    
    To prove \cref{eq:fPhicoefficients} the key is observing that for every $s,t\in \{1,2,3\}$, $x,y\in \{0,1,2,3\}$ and $a,b\in\{-1,1\}$  we have that
    \begin{equation*}
        \Tr [\sigma_{x}\ket{\chi^{s}_{a}}\langle\chi^{t}_{b}|\sigma_{y}
|\chi^{t}_{b}\rangle\langle\chi^{s}_a|]=\left\{\begin{array}{ll}
           0  & \mathrm{if}\ (s\neq x\ \mathrm{and} \ x\neq 0)\ \mathrm{or}\ (t\neq y\ \mathrm{and} \ y\neq 0),\\
            1  & \mathrm{if}\ x= 0\ \mathrm{and}\ y= 0,\\ 
            a  & \mathrm{if}\ s= x\ \ \mathrm{and}\ y= 0,\\
            b  & \mathrm{if}\ x=0\ \ \mathrm{and}\ t= y,\\
            ab  & \mathrm{if}\ s= x\ \ \mathrm{and}\ y= t.\\
        \end{array}\right.
    \end{equation*}
    Hence taking tensor products we that for every $s,t\in \{1,2,3\}^n$, $x,y\in \{0,1,2,3\}^n$ and $a=(a^1,a^2,a^3),\ b=(b^1,b^2,b^3)\in\{-1,1\}^n\times\{-1,1\}^n\times\{-1,1\}^n$ 
    it holds that 
    \begin{equation*}
        \Tr [\sigma_{x}\ket{a^s}\langle b^t|\sigma_{y}
|b^t\rangle\langle a^s|]=\langle\chi^{s}_a|\sigma_{x}\ket{\chi^{s}_{a}}\langle\chi^{t}_{b}|\sigma_{y}
|\chi^{t}_{b}\rangle=\prod_{i\in\supp x}\prod_{j\in\supp y} a_i^{x(i)}b_j^{y(j)}\delta_{x(i),s(i)}\delta_{y(j),t(j)}.
    \end{equation*}
In particular, from this follows that
\begin{align*}
    f_{\Phi_{x,y}}( a,b)&=\frac{1}{9^n}\sum_{s,t\in \{1,2,3\}^n}\Tr[\sigma_x \ket{a^s}\langle b^{t}|\sigma_y|b^{t}\rangle\langle a^{s}|]
    \\&=\frac{1}{9^n}\prod_{i\in\supp x}\prod_{j\in\supp y} a_i^{x(i)}b_j^{y(j)}\sum_{s\in \mathcal{X},t\in\mathcal{Y}} 1,
\end{align*}
where $\mathcal{X}=\{s\in\{1,2,3\}^n: s(i)=x(i)\ \forall\ i\in\supp (x)\}.$ Hence, as $|\mathcal{X}|=3^{n-|x|}$, \cref{eq:fPhicoefficients} follows for the case of $\Phi_{x,y}$. By linearity, \cref{eq:fPhicoefficients} follows for every superoperator.
\end{proof} 

\begin{theorem}[Bohnenblust-Hille inequality for $S_1\to S_\infty$ maps]\label{theo:BHchannels}
    Let $\Phi$ be a superoperator of degree at most $d$. Then, $$\norm{\widehat\Phi}_{\frac{2d}{d+1}}\leq C^{d}\norm{\Phi}_{S_1\to S_\infty}.$$ In particular, if $\Phi$ is a quantum channel, then there exists a constant $C$ such that
    $$
    \norm{\widehat\Phi}_{\frac{2d}{d+1}}\leq C^{d}.
    $$
\end{theorem}
\begin{proof}
    Let $\Re f_\Phi:\{-1,1\}^{6n}\to \mathbb R$ be defined as $(\Re f_\phi) (x)=\Re (f_\Phi (x))$ and $\Im f_\Phi:\{-1,1\}^{6n}\to \mathbb R$ as $(\Im f_\phi) (x)=\Im (f_\Phi (x))$. Note that we have that $\widehat f_{\phi}=\widehat{\Re f}_{\Phi}+i\widehat{\Im f}_{\Phi}$. By \cref{prop:propertiesfPhi}, $|(\Re f_\phi) (a,b)|,|(\Im f_\phi) (a,b)|\leq| f_\Phi (x)|\leq \norm{\Phi}_{S_1\to S_\infty}$ and that the degree of both the real and imaginary part is at most $d$. Hence, by the triangle inequality and \cref{theo:DMP} we have $$ \norm{\widehat f_{\Phi}}_{\frac{2d}{d+1}}\leq  \norm{\widehat {\Re f}_{\Phi}}_{\frac{2d}{d+1}}+ \norm{\widehat{ \Im f}_{\Phi}}_{\frac{2d}{d+1}}\leq C^{\sqrt{d\log d}}\norm{\Phi}_{S_1\to S_\infty}.$$
    Thus, using that $\norm{\widehat\Phi}_{2d/(d+1)}\leq 3^d\norm{\widehat f_{\Phi}}_{2d/(d+1)}$ it follows that $\norm{\widehat\Phi}_{2d/(d+1)}\leq C^d\norm{\Phi}_{S_1\to S_\infty}$.
    This proves the first part of the statement. 
    
    For the second we just have to show that if $\Phi$ is a quantum channel, then $\norm{\Phi}_{S_1\to S_\infty}\leq 1$. This is true since $\norm{\Phi}_{S_1\to S_\infty}\leq \norm{\Phi}_{S_1\to S_1}$ and $\Phi^\dag$ is a completely positive and unital map between C$^*$-algebras, so we have $\norm{\Phi}_{S_1\to S_1}=\norm{\Phi^\dag}_{S_\infty\to S_\infty}=1$ \cite[Proposition 3.2]{paulsenoperatoralgebras}.
    \end{proof}
\begin{remark}\label{rem:channelsimpliesunitaries}
    \cref{theo:BHchannels} generalizes the non-commutative $\BH$ inequality by Volberg and Zhang \cite{volberg2023noncommutative}, displayed in \cref{eq:BHUnitaries}. Indeed, given $M\in\mathcal{M}_N$, one can consider the superoperator $\Phi_M(\cdot)=\Id\cdot M$, which satisfies $\widehat\Phi_M(x,y)=\delta_{x,0^n}\widehat M(y)$ and $\norm{\Phi_M}_{S_1\to S_\infty}=\norm{M}_{\op}.$
\end{remark}

\subsection{Learning low-degree quantum channels}
Before we prove the main theorem of the section, we show that for a given $x,y\in \{0,1,2,3\}^n$, the corresponding Fourier coefficient $\widehat\Phi(x,y)$ can be efficiently learned. To do that, we just have to combine a few SWAP tests.
\begin{fact}[SWAP test for mixed states~\cite{kobayashi2003quantum}]\label{lem:mixedSWAPTest}
    Let $\rho, \rho'$ be two states. Then, one can estimate $\Tr[\rho\rho']$ up to error $\eps$ with probability $1-\delta$ using $O((1/\eps)^{2}\log(1/\delta))$ copies of $\rho$ and $\rho'$.
\end{fact}

\begin{lemma}[Pauli coefficient estimation for channels]\label{cor:PauliCoefficientEstimationChannel}
    Let $x,y\in \{0,1,2,3\}^n$. Then, $\widehat \Phi(x,y)$ can be estimated with error $\eps$ and probability $1-\delta$ using $O((1/\eps)^{2}\log(1/\delta))$ queries to $\Phi$. 
\end{lemma}

\begin{proof}
    If $x=y$, we just have to prepare $\widehat\Phi$ and apply \cref{lem:mixedSWAPTest} to $\widehat\Phi$ and the state $\rho=\ket{x}\bra{x}$. If $x\neq y$, one first learns $\widehat\Phi(x,x)$ and $\widehat\Phi(y,y)$ with error $\eps$ as before. One the one hand, one can learn $\widehat\Phi(x,x)+\widehat\Phi(x,x)+2\Re \widehat\Phi(x,y),$ with error $\eps$ by applying \cref{lem:mixedSWAPTest} to $\widehat\Phi$ and $1/2\sum_{z,t\in\{x,y\}}\ket{z}\bra{t}$. Hence, one learns $\Re\widehat\Phi(x,y)$ with error $3\eps/2$. On the other hand, one can learn $\widehat\Phi(x,x)+\widehat\Phi(y,y)+2\Im \widehat\Phi(x,y),$ with error $\eps$ by applying \cref{lem:mixedSWAPTest} to $\widehat\Phi$ and $1/2(\ket{x}\bra{x}+i\ket{x}\bra{y}-i\ket{y}\bra{x}+\ket{y}\bra{y}),$ and thus one can learn $\Im \widehat\Phi(x,y)$ with error $3\eps/2$. 
\end{proof}
We will also need the following well-known result about distribution learning theory. See \cite[Theorem 9]{canonne2020short} for a proof.
\begin{lemma}\label{lem:LearningDiscreteProb}
    Let $p=\{p(x)\}_x$ be a probability distribution over some set $\mathcal{X}.$ Let $p'=( p'(x))_x$ be the empirical probability distribution obtained after sampling $T$ times from $p$. Then, for $T=O((1/\eps)^2\log(1/\delta))$  with probability $\geq 1-\delta$  we have that $|p(x)- p'(x)|\leq \eps$ for every $x\in\mathcal X$.
\end{lemma}

\noindent Now, we are ready to prove \cref{theo:learnchannels}, which we restate for the convenience of the reader.
\learnchannels*

    \begin{proof} We first state the algorithm. Let $C$ be the constant in the statement of Theorem \ref{theo:BHchannels}.

    \begin{algorithm}[H]
    \caption{Learning low-degree channels via $\BH$ inequality} \label{algo:low_deg_channels}
    \textbf{Input}: A quantum channel $\Phi$ of degree at most $d$, and error $\eps$ and a failure probability $\delta$
    \begin{algorithmic}[1]
        \State Let $c=\eps^{4d+2}C^{-4d^2}$
        \State  Prepare $T_1=O((1/c)^2\log(1/\delta))$ copies of $\widehat\Phi$ to sample from $(\widehat\Phi(x,x))_x$. Let $(\widehat\Phi'(x,x))_x$ be the associated empirical distribution
        \For{$x,y\in \mathcal{X}_c=\{x:|\widehat\Phi'(x,x)|\geq c\}$}
            \State Prepare $O((1/c)^2(1/\eps)^{2}\log((1/c)^{2}(1/\delta)))$ copies of $\widehat\Phi$ and use them to approximate  $\widehat\Phi(x,y)$ with $\widehat\Phi''(x,y)$ using \cref{cor:PauliCoefficientEstimationChannel}. 
        
        \EndFor 
    \end{algorithmic}
    \textbf{Output}: $\sum_{x,y\in\mathcal{X}_c} \widehat\Phi''(x,y)\Phi_{x,y}$
\end{algorithm}
    
        Let $c>0$ to be determined later. In the first part of the algorithm we prepare $\widehat\Phi$ and measure, i.e., we sample from $(\widehat\Phi(x,x))_{x\in \{0,1,2,3\}^n}$. Let $(\widehat\Phi'(x,x))_{x\in \{0,1,2,3\}^n}$ be the empirical distribution one obtains after $T_1$ samples. Let us consider the event  $\mathcal{E}=\{|\widehat\Phi(x,x)-\widehat\Phi'(x,x)|\leq c\ \forall x\in \{0,1,2,3\}^n\}$. By \cref{lem:LearningDiscreteProb}, taking $T_1$ to be  $O((1/c)^2\log(1/\delta))$  ensures that 
\begin{equation*}
	\mathrm{Pr}[\mathcal{E}]\geq 1- \delta.
\end{equation*}
Let $\mathcal{X}_c=\{x:\ |\widehat \Phi'(x,x)|\geq c\}$. Note that, since $\sum_{x\in \mathcal{X}_c}\Phi(x,x)\leq 1,$ we know that 
\begin{align}
    |\mathcal{X}_c|\leq c^{-1}\label{eq:UpperBoundOnXa}.
\end{align}
Hence, assuming the event $\mathcal{E}$ holds, we have that
\begin{align}
    x\notin\mathcal{X}_c\implies |\widehat{\Phi}(x,x)|\leq |\widehat\Phi'(x,x)|+||\widehat \Phi(x,x)|-|\widehat\Phi'(x,x)||\leq 2c\label{eq:SnotinXa}.
\end{align}
In particular, it follows that 
\begin{align}\label{eq:SnotinXa2}
    x\notin\mathcal{X}_c\implies |\widehat{\Phi}(x,y)|\leq \sqrt{|\widehat{\Phi}(x,x)||\widehat{\Phi}(y,y)|}\leq \sqrt{2c}\quad \forall\, y\in\{0,1,2,3\}^n,
\end{align}
where in the first inequality we have used that $\widehat\Phi$ is positive semidefinite and in the second inequality \cref{eq:SnotinXa} and that $\widehat\Phi(y,y)\leq 1.$

Let us assume that the first part of the algorithm succeeds, meaning that $\mathcal{E}$ happens and we will continue the algorithm conditioned to this fact. In the second part of the algorithm we invoke  Lemma \ref{cor:PauliCoefficientEstimationChannel} to state that, by querying $\Phi$ just $$T_2=O((1/c)^4(1/\eps)^{2}\log((1/c)^{2}(1/\delta)))$$ times, we can find approximations $\widehat\Phi''(x,y)$ of $\Phi(x,y)$ for every $x,y\in\mathcal{X}_c$, such that 
$$\sup_{(x,y)\in \mathcal{X}_c\times\mathcal{X}_c}|\widehat\Phi(x,y)-\widehat\Phi''(x,y)|\leq c\eps$$happens with probability $\geq 1-\delta$.

Note that $T_2>T_1$, so this complexity dominates the one of the first part of the algorithm. Let $\widehat\Phi''(x,y)$ be these approximations and let $\Phi''=\sum_{x,y\in \mathcal{X}_c}\widehat\Phi''(x,y)\sigma_x\cdot\sigma_y$. Now, we have that 
\begin{align*}
    \norm{\Phi-\Phi''}_2^2&=\sum_{x,y\in\mathcal{X}_c}|\widehat \Phi(x,y)-\widehat\Phi''(x,y)|^2+\sum_{x\lor y\notin\mathcal{X}_c}|\widehat \Phi(x,y)|^2 \\
    &\leq\eps^2+\sum_{x\lor y\notin\mathcal{X}_c} |\widehat\Phi(x,y)|^{\frac{2}{d+1}}|\widehat \Phi(x,y)|^{\frac{2d}{d+1}}
    \\
    &\leq \eps^2+(2c)^{\frac{1}{d+1}}\norm{\widehat \Phi}_{\frac{2d}{d+1}}^{\frac{2d}{d+1}}\\
    &\leq \eps^2+c^{\frac{1}{d+1}} C^{d},
\end{align*}
where in the equality we have used Parseval's identity; in the first inequality we used~\cref{eq:UpperBoundOnXa}, the learning guarantees of the second part of the algorithm and that $2=1/(d+1/2)+2d/(d+1/2)$
; in the second inequality we have used~\cref{eq:SnotinXa2}; and in the third inequality we used the Bohnenblust-Hille inequality for channels (\cref{theo:BHchannels}). Hence, by choosing $$c=\eps^{2d+2}C^{-d(d+1)}$$ we obtain the desired result.
\end{proof}

\section{Learning low-degree quantum unitaries}
In this section our goal will be to prove the following theorem.

\learnunitaries*

The proof strategy is identical to the one of \cref{theo:learnchannels}, but now we learn the Pauli coefficients with an extension of the routine of Montanaro and Osborne (\cref{lem:est_fourier_coeff_unitary}) and use the $\BH$ ienquality of Volberg and Zhang to ensure correctness (\cref{eq:BHUnitaries}). 

The routine used in~\cite[Lemma 24]{montanaro2008quantum} for estimating Fourier coefficients of quantum Boolean functions can be extended for estimating Fourier coefficients of any quantum unitary. A Hadamard circuit is used for estimating the real and complex values of Fourier coefficients separately, which is shown in Figure~\ref{fig:qcirc_fourier_coeff_unitary}. The corresponding sample complexity for estimation can then be easily shown.
\begin{figure}[h!]
   \centering
   \includegraphics[width=0.75\textwidth]{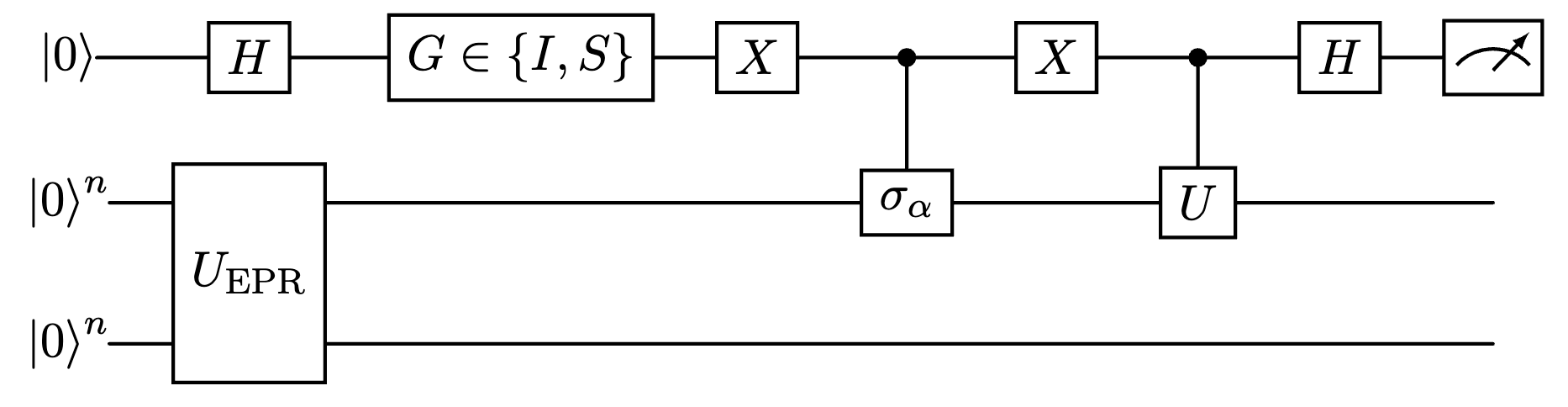}
   \caption{Quantum circuit for estimating the Fourier coefficient $\widehat{U}(\alpha)$ for any $\alpha \in \{0,1,2,3\}^n$. Gate $G$ on the ancilla qubit is set to $I$ for estimating $\mathrm{Re}(\widehat{U}(\alpha))$ and is set to the $S$ gate for estimating $\mathrm{Im}(\widehat{U}(\alpha))$. The circuit $U_{\mathrm{EPR}}$ corresponds to the state preparation unitary of $n$-qubit EPR pairs over $2n$ qubits which contains $n$ sequential $\textsf{CNOT}$ gates with control on the $j$th qubit in the first $n$-qubit register and target on the $j$th qubit in the second $n$-qubit register, ordered with increasing~$j \in [n]$.}
   \label{fig:qcirc_fourier_coeff_unitary}
\end{figure} 
\begin{lemma} \label{lem:est_fourier_coeff_unitary}
    Given a unitary $U$ and a string $x \in \{0,1,2,3\}^n$. There is an algorithm that $(\varepsilon,
    \delta)$-estimates the Fourier coefficient $\widehat{U}(x)$  using $O\left((1/\eps)^2 \log (1/\delta) \right)$ queries to $U$. 
\end{lemma}

\begin{proof}[ of Theorem~\ref{theo:learnunitaries}]
    Let $c>0$ to be determined later. In the first part of the algorithms we prepare $T_1$ copies of the  $\CJ$ state $\ket{v(U)}$ and measure them in the Bell basis, i.e., we sample from $(|\widehat U(x)|^2)_x$. Let  $(|\widehat U'(x)|^2)_{x}$ be the empirical distribution obtained after sampling. Let $\mathcal{E}=\{||\widehat{U}(x)|^2-|\widehat U'(x)|^2|\leq c^2:\, \forall\, x\in\{0,1,2,3\}^n\}$. By \cref{lem:LearningDiscreteProb}, taking $T_1=O((1/c)^4\log(1/\delta))$ times ensures that 
\begin{equation*}
	\mathrm{Pr}[\mathcal{E}]\geq 1- \delta.
\end{equation*}
Note that in the event of $\mathcal{E}$ $$||\widehat{U}(x)|-|\widehat U'(x)||\leq c,$$ for every $x\in\{0,1,2,3\}^n$. Let $\mathcal{X}_c=\{|\widehat U'(x)|\geq c\}$. Note that, since $\sum_{x\in \mathcal X_a}|\widehat U' (x)|^2\leq 1$, we have that $   |\mathcal{X}_c|\leq c^{-2}$ 
and in the event of $\mathcal{E}$ we have that
\begin{align}
    x\notin\mathcal{X}_c\implies |\widehat{U}(x)|\leq |\widehat U'(x)|+||\widehat U(x)|-|\widehat U'(x)||\leq 2c\label{eq:SnotinSa}.
\end{align}
In the second part of the algorithm we assume that the first has succeeded, i.e., that $\mathcal{E}$ occurs, Next, we learn the Pauli coefficients of $\mathcal{X}_c$ with error $c\eps$ and probability $1-\delta $ using $$T_2=O((1/c)^4(1/\eps)^2\log((1/c)^2(1/\delta)))$$ queries for each of them (which can be done using Lemma~\ref{lem:est_fourier_coeff_unitary}). As $T_2>T_1,$ the cost of this second part dominates over the first. Let $\widehat{U}''(x)$ be the approximations obtained for $\widehat U(x)$ with $x\in\{0,1,2,3\}^n$. Let $U_{c}=\sum_{x\in \mathcal{X}_c}\widehat{U}''(x)\sigma_x$. Now, we have that 
\begin{align*}
    \norm{U-U_{c}}_2^2&=\sum_{x\in\mathcal{X}_c}|\widehat U(x)-\widehat{U}'(x)|^2+\sum_{x\notin\mathcal{X}_c} |\widehat U(x)|^2\leq \eps^2+(2c)^{\frac{2}{d+1}}\norm{\widehat U}_{\frac{2d}{d+1}}^{\frac{2d}{d+1}}\leq \eps^2+c^{\frac{2}{d+1}} C^{d},
\end{align*}
where in the equality we have used Parseval; in the first inequality we have used our upper bound on $|\mathcal{X}_c|$, Eq.~\eqref{eq:SnotinSa} and the learning guarantees of the second part of the algorithm; and in the second inequality the Bohnenblust-Hille inequality of Volberg and Zhang \cref{eq:BHUnitaries}. Hence, by choosing $$c=\eps^{d+1}C^{-d(d+1)72}$$ we obtain the desired result.
\end{proof}

\section{Learning low-degree functions}

\subsection{Boolean functions.} 
In this section, we prove Facts \ref{fact:learnBool} and \ref{fact:BHBool}, both based on the well-known granurality of Boolean polynomials (see \cite[Exercise 1.11]{o2014analysis} for a proof). 

\begin{fact}[Granularity of Boolean polynomials]\label{fact:granularity}
    Let $f:\{-1,1\}^n\to \pmset{}$ be a function of degree $d$. Then, $\widehat f(s)\in 2^{1-d}\mathbb Z$ for every $s \in\{0,1\}^n.$ 
\end{fact} 
\noindent We restate the facts we are going to prove for the convenience of the reader. 
\begin{namedtheorem}[Fact 6]
Let $f:\{-1,1\}^n\to\{-1,1\}$ be a degree-$d$ function. There is a quantum algorithm that learns $f$ exactly  with probability 1$-\delta$ using $O\left(4^dd\log\left(1/\delta\right)\right)$ quantum uniform examples. In contrast, any classical algorithm requires $\Omega(2^d+\log n)$ uniform examples and there is a classical algorithm that only uses $O\left(4^dd\log\left(n/\delta\right)\right)$ uniform examples.
\end{namedtheorem}
\begin{proof}
    The classical lower bound follows from the fact that $d$-juntas (functions depending only on $d$ variables) have degree $d$ and it requires $\Omega(2^d+\log n)$ to learn them (see, for instance, \cite[Lemma IV.1]{atici2007quantum}).
     For the classical upper bound we take $T=O(4^d\log(n^d/\delta))$ uniform samples $(x_i,f(x_i))$. Define the empirical Fourier coefficients as $$\widehat f'(s)=\frac{1}{T}\sum_{i\in [T]}f(x_i)\chi_s(x_i).$$
    Define the event $\mathcal{E}=\{|\widehat f(s)-\widehat f'(s)|<2^{-d}\ \forall\, |s|\leq d\}.$
    From \cref{eq:ParsevalIdentity}, the Hoeffding bound and a union bound over the $\sum_{i\in [d]_0}{n\choose i}\leq n^d$ strings $s$ with $|s|\leq d$ it follows that $$\mathrm{Pr}[\mathcal{E}]\geq 1-\sum_{i\in [d]_0}{n\choose i}e^{-\frac{4^{-d}T}{2}}\geq 1-\delta.$$
    In the second part of the algorithm, we round every $\widehat f'(s)$ to the closest number $\widehat f''(s)\in 2^{d-1}\mathbb Z$. If $\mathcal{E}$ occurs, by granularity we have that $\widehat f''(s)=\widehat f(s)$ for every $|s|\leq d,$ so we learn $f$ exactly. 
    
    For the quantum upper bound we first use $T_1=O(4^dd\log^2(1/\delta))$ quantum uniform samples to do Fourier sampling (see for instance~\cite[Lemma 4]{arunachalam2021two}). This allows us to sample $4^d\log(4^d/\delta)$ times from $(\widehat f(s)^2)_{s\in [n]}$. 
    Given $s$ such that $\widehat f(s)\neq 0$, and thus $|\widehat f(s)|^2\geq 4^{1-d}$ by \cref{fact:granularity}, the probability of not sampling $s$ is at most 
    $$
    \left(1-4^{1-d}\right)^N\leq {\delta}/{ 4^d}.
    $$
    Then, taking a union bound over the at most $4^{d-1}$ non-zero Fourier coefficients (due to \cref{fact:granularity} and $\sum_s|\widehat f(s)|^2=1$), it follows that with probability $1-\delta$ we will have sampled every non-zero Fourier coefficient. In the second part of the algorithm we use $T_2=O(4^d\log(4^d/\delta))$ quantum uniform samples and measure them in the computational basis, which generates \emph{classical} uniform samples. From here, we can argue as in the classical upper bound and learn $f$ exactly. The quantum advantage comes from Fourier sampling, that allows to detect the relevant Fourier coefficients, after what one can learn those few coefficients in a classical way.

    Note that the  quantum advantage comes from Fourier sampling, that allows to detect the relevant Fourier coefficients, after what one can learn those few coefficients in a classical way.
\end{proof}

\begin{namedtheorem}[Fact 7]
    Let $f:\{-1,1\}^n\to\{-1,1\}$ of degree at most $d$. Then, 
    $$
    \Big(\sum_{s\in\{0,1\}^n}|\widehat f(s)|^{\frac{2d}{d+1}}\Big)^{\frac{d+1}{2d}}\leq  2^{\frac{d-1}{d}}.$$
    The equality is witnessed by the address function.  
    
\end{namedtheorem}
\begin{proof}
    Granularity (\cref{fact:granularity}) and  Parseval's identity \cref{eq:ParsevalIdentity}, $\sum_{s}\widehat f^2(s)=\mathbb E_x f(x)^2=1,$  imply that there $f$ has at most $2^{2(d-1)}$ non-zero Fourier coefficients. Hence, using H\"older's inequality it follows that for $p\in [1,2)$
    \begin{align*}
        \sum_{s:\widehat f(s)\neq 0}|\widehat f|^p\cdot 1\leq \left(\sum_{s\in\{0,1\}^n}\widehat{f}^2(s)\right)^{\frac{p}{2}}\left(2^{2(d-1)}\right)^{\frac{2-p}{2}}=2^{(d-1)(2-p)}.
    \end{align*}
    Taking $p=2d/(d+1)$ the claimed inequality follow. 
    
    The equality is witnessed by the \emph{address function} $f:(\{-1,1\})^{n})^d\to \{-1,1\}$ of degree $d$ and $n=2^{d-1}$, which is defined as 
    \begin{equation}
        f(x)=\sum_{a\in \{-1,1\}^{d-1}}\underbrace{\frac{x_1(1)-a_1x_1(2)}{2}\dots \frac{x_{d-1}(1)-a_{d-1}x_{d-1}(2)}{2}}_{g_a(x_1,\dots,x_{d-1})}x_{d}(a),
    \end{equation}
    where we identify $\{-1,1\}^{d-1}$ with $[2^{d-1}]$ in the canonical way. The address function is Boolean because for every $(x_1,\dots,x_{d-1})\in (\{-1,1\}^{n})^{d-1}$ there is only one $a\in \{-1,1\}^{d-1}$ such that $g_a(x_1,\dots,x_{d-1})$ is not $0$, in which case it takes the value $\pm 1$. Given that it has $2^{2(d-1)}$ Fourier coefficients and all of them equal $2^{1-d}$, we have that $$\Big(\sum_{s\in\{0,1\}^n}|\widehat f(s)|^{\frac{2d}{d+1}}\Big)^{\frac{d+1}{2d}}=2^{1-d}2^{2(d-1)\cdot \frac{d+1}{2d}}=2^{\frac{d-1}{d}},$$
    as promised.
\end{proof}

\subsection{Bounded functions.}

In this section our goal will be to prove the following proposition.

\begin{namedtheorem}[Proposition 8]
    Let $p:\{-1,1\}^n\to [-1,1]$ be a degree $d$ polynomial. There is an algorithm that $(\varepsilon,\delta)$-learns $p$ (in $\ell_2$-accuracy) using $
    \exp(\widetilde{O}(\sqrt{d^{3}\log d}+d\log(1/\eps))\cdot \log (1/\delta))$ copies of a block-encoding of $p$.
\end{namedtheorem}
In comparison to the Boolean function scenario, for bounded multilinear polynomials $p:\{-1,1\}^n\to [-1,1]$ things get more complicated because \cref{fact:granularity} no longer holds and uniform random samples do not make sense. However, we can circumvent both issues. The first, because, as noted by Eskenazis and Ivanisvili, it is already useful to have an upper bound for $\norm{\widehat p}_q$ for a $q<2$~\cite{eskenazis2022learning}, which is exactly what the Bohnenblust-Hille inequality of Defant et al. (\cref{theo:DMP}). The second issue can be solved by accessing bounded functions $p:
\{-1,1\}^n\to [-1,1]$ through block-encodings in $2^{n+1}$-dimensional unitary $U_p$. We consider the block-encodings defined by
\begin{equation}
\label{eq:blockenc}
    U_p=\begin{pmatrix} \Diag(p) & -\Diag((1-p^2)i)\\ \Diag((1-p^2)i)& -\Diag(p)\end{pmatrix},
\end{equation}
where for a function $g:\{-1,1\}^n\to \mathbb C$ we define $\Diag g$ as the $2^n\times 2^n$ diagonal matrix whose diagonal is given by the values of $g$. In terms of the Pauli matrices and the Fourier spectrum of $p$, $U_p$ can be written as 
\begin{equation}
\label{eq:blockencpau}
    U_p= \sigma_3\otimes \left[\sum_{s\in \{0,1\}}\widehat p(S) \otimes_{i\in \supp (s)}(\sigma_3)_i\right]+\sigma_2\otimes \left[\sum_{s\in\{0,1\}^n}\widehat{1-p^2}(s) \otimes_{i\in \supp(s)}(\sigma_3)_i\right],
\end{equation}
where $\otimes_{i\in \supp (s)}(\sigma_3)_i$ stands for the Pauli string with $\sigma_3$ indicated by $s$ and identities on the rest.  Now, we are ready to present our quantum learning algorithm for bounded functions, which combines Fourier sampling techniques and the learning ideas of Eskenazis and Ivanisvili~\cite{eskenazis2022learning}.

\begin{proof}[ of \cref{prop:learnBlockEncoding}]
    Let $a,b>0$ to be fixed later. In the first of the algorithm we prepare $T_1=(1/a)^{2}\log(2/\delta a^2)$ copies of the  $\CJ$ state of $U_p$ and measure them in the basis $\{\ket{v(\sigma_x)}\}$.
    By \cref{eq:CJBellExpansion}, this is the same as sampling $T_1$ times from $(|\widehat U_p(x)|^2)_{x\in\{0,1,2,3\}^{n+1}}$. Now note that $$\widehat U_p(x)=\widehat p( s_x)$$
    if $ x\in \{3\}\times \{0,3\}^n$, where $s_x\in \{0,1\}^n$ is defined as $(s_x)_i:=\delta_{x_{i+1},3}$. Let $\mathcal{S}_a=\{ s\in \01^n: |\widehat p( s)|\geq a\}$. The probability of not measuring an $s\in \mathcal{S}_a$ in $T_1$ samples is at most $$(1-a^2)^{T_1}\leq a^2 \delta.$$ Thus, as $|\mathcal{S}_a|\leq a^{-2}$ (because $\sum_{s}|\widehat p(s)|^2= \norm{p}_2^2\leq 1$), it follows that after $T_1$ samples of $(|\widehat U_p(x)|^2)_x$ with probability$\geq 1-\delta$ we will have sampled at least once every element of $\mathcal{S}_a$. In other words, if $\mathcal{T}_1$ are the $s\in \{0,1\}^n$ such that correspond to an $x\in\{0,1,2,3\}^{n+1}$ that has been sampled at least once, then $\mathcal{S}_a\subseteq \mathcal{T}_1$ with high probability. 

    In the second part of the algorithm we use $$T_2=O((1/b)^2\log(2/(\delta T_1)))$$ uniform random samples to learn with error $b$ and probability $\geq 1-\delta$ the at most $T_1$ Fourier coefficients sampled in the first part of the algorithm. This can be done by estimating the empirical Fourier coefficients, as in the proof of \cref{fact:learnBool}. We denote these estimations by $\widehat p'(s).$ Let $h_{a,b}=\sum_{s\in\mathcal{T}_1}\widehat p'(s)\chi_s.$ We claim that $h_{a,b}$ is a good approximation of $p$. Indeed, 
    \begin{align*}
        \norm{p-h_{a,b}}_2^2&= \sum_{s\in \mathcal{T}_1}|\widehat p(s)-\widehat p'(s)|^2+\sum_{s\notin \mathcal{T}_1}|\widehat p(s)|^2\\
        &\leq b^2T_1+\sum_{s\notin \mathcal{T}_1}|\widehat p(s)|^{\frac{2d}{d+1}}|\widehat p(s)|^{\frac{2}{d+1}}\\
        &\leq b^2\frac{1}{a^2}\log\left(\frac{2}{\delta a^2}\right)+a^{\frac{2}{d+1}}\norm{\widehat p}_{\frac{2d}{d+1}}^{\frac{2d}{d+1}}\\
        &\leq b^2\frac{1}{a^2}\log\left(\frac{2}{\delta a^2}\right)+a^{\frac{2}{d+1}}C^{\sqrt{d\log d}},
    \end{align*}
    where in the equality we have used Parseval's identity; in the first inequality we have used the learning guarantee of the second part; in the second inequality that $S_a\subseteq\mathcal{T}_1$; and in the third inequality the $\BH$ inequality for multilinear polynomials (\cref{theo:DMP}). 
    Hence, taking $b^2=\eps^2a^2(\log(2/\delta a^2))^{-1}$ and $a=\eps^{d+1}C^{d^{3/2}(\log (d))^{1/2}}$ it follows that $\norm{p-h_{a,b}}\leq \eps^2,$ as desired. This way, $T_2>T_1$ and the sample complexity is upper bounded by $$T_2=\widetilde O((1/b^2)\log(1/\delta))=\widetilde O((1/a^2\eps^2)\log(1/\delta))=\widetilde O((1/\eps)^{2d+4}C^{d^{3/2}(\log (d))^{1/2}}\log(1/\delta)).$$

    Regarding the time complexity dependence on $n$, the algorithm runs in $\poly(n)$ up to $d=(\log n)^{\eta}$ with $\eta<2/3$, because generating a $\CJ$ state requires $O(n)$ elementary gates (we generate $T_1$ $\CJ$ states); computing a empirical Fourier coefficient requires $\poly(n)$ time (we compute $T_1$ empirical Fourier coefficients); and $T_1=\widetilde O(\eps^{2d+2}C^{ d^{3/2}(\log(d))^{1/2} })$, so $T_1=O(\poly (n))$ even for $d=(\log n)^{\eta}.$
\end{proof}

\section{Learning quantum query algorithms}
In this section our goal will be to prove the \cref{theo:learnqpolynomials}, restated below for the readers convenience.
\learnqpolynomials*
\begin{proof}
    Eskenazis and Ivanishvili showed that a function $f:\{-1,1\}^n\to [-1,1]$ with degree at most $d$ and $\norm{\widehat f}_{\frac{2d}{d+1}}\leq C$, can be learned with success probablity $1-\delta$ and error $\eps$ in $\ell_2^2$ accuracy using 
    \begin{equation}\label{eq:EskenazisLearn}
        O(\eps^{-(d+1)}C^{2d}\log(n/\delta)) 
    \end{equation}
    samples $(x,f(x))$, where $x$ is drawn uniformly at random from $x\in\{-1,1\}^n$.     The first author, Bri\"et and Palazuelos~\cite{QQA=CBF} showed that the amplitudes of quantum algorithms that $d$ queries as in \cref{fig:QQA} are completely bounded $d$-tensors. Hence, using \cref{theo:BHGf} (which we prove below) one can use \cref{eq:EskenazisLearn} with $C=1$, from which the stated result follows.
\end{proof}

\paragraph{The constant of the completely bounded $\BH$ inequality is 1.} In this section we determine that the exact value of the constant of the completely bounded $\BH$ inequality is 1.
\begin{theorem}\label{theo:BHGf}
    Let $\mathbb K\in\{\mathbb R,\mathbb C\}$. Let $T:(\mathbb K^n)^{d}\to\mathbb K$ be a $d$-linear form. Then, 
    $$\norm{\widehat T}_{\frac{2d}{d+1}}\leq \norm{T}_{\cb},$$ and the equality can be witnessed. 
\end{theorem}
\cref{theo:BHGf} establishes that the best constant for the completely bounded $\BH$ inequality is exactly $1$. This sharply contrasts with the current knowledge about the Grothendieck constants, and specially with the $\BH$ constants, where only poly($d$) upper bounds are known. This way, we \emph{close} one of the edges of comparison of the three norms that appear in Grothendieck and Bohnenblust-Hille inequalities (see  \cref{fig:normcomparisons}).

\begin{figure}[ht]
    \centering
    \hfill 
    \begin{minipage}{0.45\textwidth}
        \begin{tikzpicture}
    \node (top) at (0.1, 2) {$\lVert B \rVert_{\infty}$};
    \node (bottomLeft) at (-1.5, 0) {$\lVert B\rVert_{\frac{2d}{d+1}}$};
    \node (bottomRight) at (1.8, 0) {$\lVert B \rVert_{\cb}$};

    \node[rotate=45] (top) at (-1, 1) {$\leq  \sqrt{2}$};
    \node[rotate=-45] (top) at (1, 1) {$G_{\mathbb R}\geq $};
    \node (top) at (0,0) {$\leq 1$};
\end{tikzpicture}
    \end{minipage}
    \hfill
    \begin{minipage}{0.45\textwidth}
            \begin{tikzpicture}
        \node (top) at (0.1, 2) {$\lVert T \rVert_{\infty}$};
        \node (bottomLeft) at (-1.5, 0) {$\lVert T \rVert_{\frac{2d}{d+1}}$};
        \node (bottomRight) at (1.8, 0) {$\lVert T \rVert_{\cb}$};
    
        \node[rotate=45] (top) at (-1, 1.1) {$\leq \poly (d)$};
        \node[rotate=-45] (top) at (1, 1) {$C(d)\not\geq $};
        \node (top) at (0,0) {$\leq 1$};
    \end{tikzpicture}
    \end{minipage}
    \caption{Triangles of norm comparisons. In the left triangle we display the norm comparisons implied by the Littlewood~\cite{littlewood1930bounded} and Grothendieck inequalities~\cite{grothendieck1953resume} and our \cref{theo:BHGf} for real bilinear maps. In the right triangle we depict the best upper bound for the $\BH$ constant~\cite{bayart2014bohr}, the no extension of the Grothendieck inequality~\cite{perez2008unbounded}, and our \cref{theo:BHGf} for $d$-linear~tensors.}\label{fig:normcomparisons}
\end{figure}
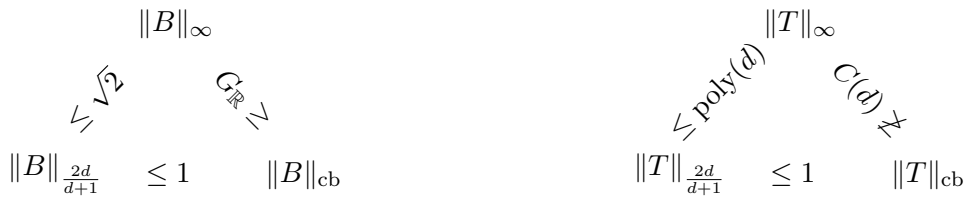
The main ingredient of the proof of \cref{theo:BHGf} is a general lower bound for the completely bounded norm, \cref{lem:lowercb}. This technique is inspired on an idea of Varoupoulos to rule out a generalization of von Neumann's inequality~\cite{Varopoulos:1974} and was already used by the third author to show a particular case of the famous Aaronson and Ambainis conjecture of quantum query complexity~\cite{gutierrez2023influences}. Combining \cref{lem:lowercb} with Blei's inequality (\cref{lem:Bleisinequality}), \cref{theo:BHGf} follows. For a proof of Blei's inequality see~\cite[Theorem 2.1]{bayart2014bohr}.

\begin{lemma}[Blei's inequality] \label{lem:Bleisinequality}
    Let $\mathbb K\in\{\mathbb R,\mathbb C\}$ and let $\widehat T\in(\mathbb K)^{n^d}$. Then, 
    $$\left(\prod_{s\in[d]}\sum_{i_s\in [n]}\sqrt{\sum_{i_1,\dots,i_{s-1},i_{s+1},\dots,i_d\in [n]}|\widehat T_\ind{i}|^2}\right)^{\frac{1}{d}}\geq \norm{\widehat T}_{\frac{2d}{d+1}}.$$
\end{lemma}
\begin{lemma}\label{lem:lowercb}
    Let $\mathbb K\in\{\mathbb R,\mathbb C\}$, let $T:(\mathbb K^n)^{d}\to \mathbb K$ be a $d$-linear form, and let $s\in [d]$. Then, 
    $$\norm{T}_{\cb}\geq \sum_{i_s\in [n]}\sqrt{\sum_{i_1,\dots,i_{s-1},i_{s+1},\dots,i_d\in [n]}|\widehat T_\ind{i}|^2}.$$
\end{lemma}
\begin{proof}
    Let us fix $s\in [d]$. The proof consists on evaluating \cref{eq:Tcb} on an explicit set of contractions (matrices with operator norm at most 1). Let $m=\sum_{r=0}^{d-r}n^{r}+\sum_{r=0}^{s-1}n^{r}$ and let $\{e_{\ind i},f_{\ind j}:\, \ind i\in [n]^r,\, r\in\{0\}\cup [d-s],\, \ind j\in [n]^t,\, t\in\{0\}\cup [s-1]\}$ be an orthonormal basis of  $\ell_2^m (\mathbb K),$ where we identify $[n]^0$ with $\emptyset$. We define matrices $X(i)\in\mathcal{M}_m$ for $i\in [n]$  by 
    \begin{align*}
        X(i)e_{\ind j}&=e_{(i,\ind j)},\ \mathrm{if\ }\ind j\in [n]^{r},\ r\in \{0\}\cup [d-s-1],\\
        X(i)e_{\ind j}&=\frac{\sum_{\ind k\in [n]^{s-1}}\widehat T_{\ind k i\ind j}^*f_{\ind k}}{\sqrt{\sum_{k_1,\dots,k_{s-1},k_{s+1},\dots,k_d\in [n] }|\widehat T|^2_{(k_1,\dots,k_{s-1},i,k_{s+1},\dots,k_d)}}},\ \mathrm{if\ }\ind j\in [n]^{d-s},\\
        X(i)f_{\ind j}&=\delta_{i,j_{t}}f_{(j_1,\dots,j_{t-1})},\ \mathrm{if}\ \ind j\in [n]^{d-t},\ t\in\{0\}\cup[d-1].
    \end{align*}
    To give some intuition of the behaviour of these matrices, one may interpret the first $d-s$ applications of the matrices $X(i)$ as \emph{creation} operators and the last $s-1$ as \emph{destruction} operators.
    
    Assume for the moment that $X(i)$ are contractions. Given that  
    \begin{align*}
        \langle f_\emptyset,X(i_1)\dots X(i_d)e_\emptyset\rangle=\frac{\widehat T_\ind i^*}{\sqrt{\sum_{ i_1,\dots,i_{s-1},i_{s+1},\dots,i_d\in [n]}|\widehat T_{\ind i}}|^2},
    \end{align*}
    it would then follow that 
    \begin{align*}
        \norm{T}_{\cb}&\geq \norm{\sum_{\ind i\in [n]^d}\widehat T_{\ind i}X(i_1)\dots X(i_d)}_{\op}\\
        &\geq  \sum_{\ind i\in [n]^d}\widehat T_{\ind i}\langle f_\emptyset,X(i_1)\dots X(i_d)e_\emptyset\rangle\\
        &\geq \sum_{\ind i\in [n]^d}\widehat T_{\ind i}\frac{\widehat T_\ind i^*}{\sqrt{\sum_{ k_1,\dots,k_{s-1},k_{s+1},\dots,k_d\in [n] }|\widehat T|^2_{(k_1,\dots,k_{s-1},i_s,k_{s+1},\dots,k_d)}}}\\        &=\sum_{i_s\in [n]}\sqrt{\sum_{i_1,\dots,i_{s-1},i_{s+1},\dots,i_d\in [n]}|\widehat T_\ind{i}|^2},
    \end{align*}
    as desired. Thus, it just remains to prove that the matrices $X(i)$ are contractions. Given that $X(i)$ maps $\{e_{\ind i}:\, \ind i\in [n]^r,\, r\in\{0\}\cup [d-s-1]\}$, $\{e_{\ind i}:\, \ind i\in [n]^{d-s}\}$ and $\{f_{\ind i}:\, \ind i\in[n]^t,\, t\in\{0\}\cup [s]\}$ to orthogonal subspaces, it suffices to show that the $X(i)$ are contractions when restricted to those subspaces. For the first and third sets that is true because  $X(i)$ maps each basis vector of those sets to a different basis vector or to $0$. For the second set, is also true because for every $\lambda \in \mathbb K^{n^{d-s}}$
    \begin{align*}
        \norm{X(i)\sum_{\ind j\in [n]^{d-s}}\lambda_{\ind j}e_{\ind j}}_2^2&=\norm{\frac{\sum_{\ind k\in [n]^{s-1}}(\sum_{\ind j\in [n]^{d-s}}\lambda_{\ind j}\widehat T_{\ind k i\ind j}^*)f_{\ind k}}{\sqrt{\sum_{k_1,\dots,k_{s-1},k_{s+1},\dots,k_d\in [n] }|\widehat T_{k_1,\dots,k_{s-1},i,k_{s+1},\dots,k_d}|^2}}}_2^2\\
        &=\frac{\sum_{\ind k\in [n]^{s-1}}|\sum_{\ind j\in [n]^{d-s}}\lambda_{\ind j}\widehat T_{\ind k i\ind j}^*|^2}{\sum_{k_1,\dots,k_{s-1},k_{s+1},\dots,k_d\in [n] }|\widehat T_{k_1,\dots,k_{s-1},i,k_{s+1},\dots,k_d}|^2}\\
        &\leq \frac{\Big(\sum_{\ind k\in [n]^{s-1}}\sum_{\ind j\in [n]^{d-s}}|\widehat T_{\ind k i\ind j}^*|^2\Big)\Big(\sum_{\ind j\in [n]^{d-s}}|\lambda_{\ind j}|^2\Big)}{\sum_{k_1,\dots,k_{s-1},k_{s+1},\dots,k_d\in [n] }|\widehat T_{k_1,\dots,k_{s-1},i,k_{s+1},\dots,k_d}|^2}\\
        &=\sum_{\ind j\in [n]^{d-s}}|\lambda_{\ind j}|^2=\norm{\lambda}^2_2,
    \end{align*}
    where  we have used Cauchy-Schwarz for the sum over $\ind j$.
\end{proof}

\begin{proof}[ of \cref{theo:BHGf}]
    The inequality $\|\widehat{T}\|_{\frac{2d}{d+1}}\leq\norm{T}_{\cb}$ follows from \cref{lem:Bleisinequality,lem:lowercb}. The inequality is saturated by the form $T(x_1,\dots,x_d)=x_1(1).$ 
\end{proof}

\bibliographystyle{alphaurl}
\bibliography{Bibliography}

\appendix
\section{Learning low-degree Pauli channels}
In this section, we consider learning low-degree Pauli channels, a special case of general quantum channels and prove the following theorem.
\learnPauliChannels*

We define the orthonormal basis for Pauli channels as follows. For any $x \in \{0,1,2,3\}^n$, let $\Phi_x = \sigma_x \rho \sigma_x$. For Pauli channels $\Phi$, the Pauli expansion is then
\begin{equation}
    \Phi = \sum_{x \in \{0,1,2,3\}^n} \widehat{\Phi}(x) \Phi_x,
\end{equation}
where $\widehat{\Phi}(x)$ are the Fourier coefficients corresponding to the characters $\Phi_x$. It can then be immediately shown that the Fourier coefficients $\widehat{\Phi}(x)$ are nothing but the error rates $p_x$. Our goal is learning a classical description of the Pauli channel $\Phi$ within $\eps$-diamond norm. The diamond norm between two Pauli channels is equivalent to twice the total variation distance between their corresponding probability distributions (or error rates) \cite{magesan2012character}. We thus aim to learn the probability distribution $\widehat{\Phi} = \{ \widehat{\Phi}(x) \}_x$ over the error rates of a Pauli channel within $\eps/2$ total variation distance.

The eigenvalues of a Pauli channel are given by $\lambda_s = \frac{1}{2^n} \Tr(\sigma_s \Phi(\sigma_s))$. They are related to the error rates by a Fourier transform:
\begin{align}
    \lambda_k = \sum_{j \in \{0,1,2,3\}^n} (-1)^{[\sigma_k, \sigma_j]} p_j, \quad
    p_k = \frac{1}{4^n} \sum_{j \in \{0,1,2,3\}^n} (-1)^{[\sigma_k, \sigma_j]} \lambda_j
\end{align}
Most approaches for learning Pauli channels~\cite{flammia2020efficient,harper2021fast} which have prior information regarding the sparsity in $\widehat{\Phi}$ have attempted to learn the channel by estimating the channel's eigenvalues. In contrast, we directly estimate error rates.

Let us warm-up by giving a learning algorithm for low-degree Pauli channels when given access to entanglement.
\begin{theorem}
    Let $\Phi$ be an $n$-qubit degree-$d$ channel. There is an algorithm that $(\varepsilon,
    \delta)$-learns $\Phi$ (in  diamond norm)  using $O\left( n^d/{\varepsilon^2} + 1/\varepsilon^2\cdot \log({1}/{\delta})\right)$ queries to $\Phi$.
\end{theorem}
\begin{proof}
    Let $J(\Phi)$ be the CJ state of $\Phi$. If we measure $J(\Phi)$ in the Bell basis, we sample $x \in \{0,1,2,3\}^n$ with probability $\widehat{\Phi}(x)$ as long as $|x| \leq d$. By~\cite[Theorem 1]{canonne2020short}, measuring $J(\Phi)$, $O\left(\frac{3^d n^d}{\varepsilon^2} + \frac{1}{\varepsilon^2}\log\frac{1}{\delta} \right)$ times we learn the corresponding probability distribution $\{\widehat{\Phi}(x)\}_x$ of the Pauli channel in total variation distance of $\varepsilon$ with probability $1-\delta$. This translates to learning the corresponding Pauli channel in $\varepsilon$-diamond norm.
\end{proof}

We now describe our learning algorithm for low-degree Pauli channels that does not require any access to entanglement and can be implemented on near-term quantum devices. As part of the learning algorithm, we apply the Pauli channel $\Phi$ to a product state $\rho$ and then measure $\Phi(\rho)$ in the Pauli basis. Let us describe what these product states, measurements and the resulting outcomes are. Let us first recall from Section~\ref{sec:prelims} that we denote the $\pm 1$ eigenstates of a single-qubit Pauli $\sigma_j$ (where $j \in \{1,2,3\}$) by $|\chi_{\pm 1}^{j}\rangle$. The algorithm repeats the following procedure $T$ times: pick a uniformly random $s
\in \{1,2,3\}^n$, prepare the product state $\rho_s = \otimes_{j=1}^s \ket{\chi_{+1}^{s_j}}\bra{\chi_{+1}^{s_j}}$, apply $\Phi$, measure $\Phi(\rho_s)$ in the Pauli basis $\sigma_s$, and store the outcome $r$. The remaining of the section is devoted to show that with a few repetitions of this procedure, one can recover $\Phi$ from the outcomes $r$. The first thing to do is to determine the behaviour of $r$ as random variable. To express it in a concise way, we introduce the following `$\star$' operation. Given $s,x\in\{0,1,2,3\}$, $(s\star x)$ is the element of $\{0,1\}^n$ defined by 
\begin{equation}\label{eq:stardefinition}
    (s\star x)_j=\left\{\begin{array}{ll}
         0 &
        \mathrm{if}\ [\sigma_{s_j},\sigma_{x_j}]=0,  \\
         1 & \mathrm{otherwise},
    \end{array}\right.
    \quad \forall j \in [n].
\end{equation}

\begin{proposition}
Let $\Phi_x(\rho) = \sigma_x \rho \sigma_x$ be the superoperator corresponding to $x \in \{0,1,2,3\}^n$ and let $s\in \{1,2,3\}^n.$ Measuring $\Phi_x(\rho_s)$ in the Pauli basis $\sigma_s$ produces the measurement outcome $r = s \star x$. 
\label{prop:outcomes_queries_pauli_channel}
\end{proposition}

\begin{proof}

    Upon measuring the $j$th qubit of the state $\sigma_x \rho_s \sigma_x$ in basis $\sigma_s$, the probability of the measurement outcome $r_j \in \{0,1\}$ is given by
    \begin{align*}
        \Pr[r_j = 0] &= \frac{1}{2} + \frac{1}{2}\Tr\left(\sigma_{s_j} \sigma_{x_j} \ket{\chi_{+1}^{s_j}}\bra{\chi_{+1}^{s_j}} \sigma_{x_j} \right) \\ 
        &= \frac{1}{2} + \frac{1}{2}\Tr\left(\sigma_{x_j} \sigma_{s_j} \sigma_{x_j} \ket{\chi_{+1}^{s_j}}\bra{\chi_{+1}^{s_j}} \right) \\
        &= \frac{1}{2} + \frac{1}{2}(-1)^{[\sigma_{s_j}, \sigma_{x_j}]} \Tr\left(\sigma_{s_j} \ket{\chi_{+1}^{s_j}}\bra{\chi_{+1}^{s_j}} \right) \\
        &= \frac{1 + (-1)^{[\sigma_{s_j}, \sigma_{x_j}]}}{2},
    \end{align*}
    where we used the fact $\sigma_{a} \sigma_{b} \sigma_{a} = (-1)^{[\sigma_a,\sigma_b]} \sigma_b$ for any $a,b \in \{0,1,2,3\}^n$ in the third line and $\ket{\chi_{+1}^{s_j}}$ is the $+1$ eigenstate of $\sigma_{s_j}$. Note from the above probability expression, we have that $r_j$ takes the value of $0$ if $\sigma_{s_j}$ and $\sigma_{x_j}$ commute and $1$, otherwise. This proves the desired result.
\end{proof}

Let $\{(s^{(t)}, r^{(t)})\}_{t \in [T]}$ be the samples obtained by our algorithm. We now describe how to learn the Pauli coefficients (or error rates) $\widehat{\Phi}(x)$ from these $T$ samples, and thereby learn the Pauli channel. The following empirical estimation gives us an estimate of the Pauli coefficients, which we denote by $\widetilde{\Phi}(x)$
\begin{equation}
    \widetilde{\Phi}(x) = \frac{1}{T} \sum_{t=1}^T \left[ \left(- \frac{1}{2}\right)^{\left(\sum \limits_{j \in \text{supp}(x)} r^{(t)}_j \oplus ({s^{(t)}_j} \star x_j) \right) \oplus \left(\sum \limits_{j \notin \text{supp}(x)} r^{(t)}_j \right) } \right],
    \label{eq:emp_estimate_phix_pauli_channel}
\end{equation}
where we have used $\oplus$ to denote addition modulo $2$, the operation `$\star$' which was defined in Prop.~\ref{prop:outcomes_queries_pauli_channel} and $\text{supp}(x)$ denotes the support of the string $x$ over non-zero values. Note that the internal summation involving the Fourier character $x$ is now only over its support which is at most of size $d$. The empirical estimator of Eq.~\ref{eq:emp_estimate_phix_pauli_channel} was proposed in \cite{flammia2021pauli}. The following statement argues that the empirical estimate of Eq.~\ref{eq:emp_estimate_phix_pauli_channel} is indeed consistent (which was previously shown as part of \cite[Theorem 22]{flammia2021pauli} but is described here for completeness).

\begin{proposition} For any $x \in \{0,1,2,3\}^n$, we have $\mathop{\mathbb{E}}_{s \sim \{1,2,3\}^n}\left[ \widetilde{\Phi}(x) \right] = \widehat{\Phi}(x).$
\label{prop:consistency_emp_est_pauli_channels}
\end{proposition}
\begin{proof}
    Considering the expectation over the empirical estimate and substituting Eq.~\ref{eq:emp_estimate_phix_pauli_channel}, we have
    \begin{align}
        \mathop{\mathbb{E}}_{s \sim \{1,2,3\}^n}\left[ \widetilde{\Phi}(x) \right] &= \mathop{\mathbb{E}}_{s \in \{1,2,3\}^n}\left[ \left(- \frac{1}{2}\right)^{\left(\sum \limits_{j \in \text{supp}(x)} r_j \oplus ({s_j} \star x_j) \right) \oplus \left(\sum \limits_{j \notin \text{supp}(x)} r_j \right) } \right] \nonumber \\
        &= \mathop{\mathbb{E}}_{s \sim \{1,2,3\}^n}\left[ \left(- \frac{1}{2}\right)^{\sum \limits_{j \in [n]} r_j \oplus ({s_j} \star x_j)} \right] \nonumber \\
        &= \mathop{\mathbb{E}}_{z \sim \widehat{\Phi}} \mathop{\mathbb{E}}_{s \sim \{1,2,3\}^n} \left[ \left(- \frac{1}{2}\right)^{\sum \limits_{j \in [n]} ({s_j} \star z_j) \oplus ({s_j} \star x_j)} \mid z \right]
        \nonumber \\
        &= \mathop{\mathbb{E}}_{z \sim \widehat{\Phi}} \prod_{j\in [n]}\mathop{\mathbb{E}}_{s \sim \{1,2,3\}^n} \left[ \left(- \frac{1}{2}\right)^{ ({s_j} \star z_j) \oplus ({s_j} \star x_j)} \mid z \right],
        \label{eq:simplified_emp_est}
    \end{align}
where we simplified the summation in the second line noting that $x_j = 0$ for $j \notin \text{supp}(x)$. In the third line, we used the law of total expectations considering that the action of a Pauli channels involves the action of the superoperator $\Phi_z$ on the input state with probability $\widehat{\Phi}(z)$ and which allows us to then use Prop.~\ref{prop:outcomes_queries_pauli_channel} regarding the measurement outcome $r$. Note that for a given $z$ and $j \in [n]$,
\begin{equation}
    \mathop{\mathbb{E}}_{s_j \sim \{1,2,3\}} \left[ \left(- \frac{1}{2}\right)^{({s_j} \star z_j) \oplus (s_j\star x_j)} \mid z \right] = 
    \begin{cases}
        1 & \text{if } z_j = x_j, \\
        0 & \text{if } z_j \neq x_j, \\
    \end{cases}
    \label{eq:expectation_emp_est_term}
\end{equation}
which can be readily verified.
The expression of Eq.~\ref{eq:expectation_emp_est_term} can thus be simply written as $\mathds{1}\{z_j = x_j\}$. Substituting back into Eq.~\ref{eq:simplified_emp_est} gives~us
\begin{equation}
    \mathop{\mathbb{E}}_{s \sim \{1,2,3\}^n}\left[ \widetilde{\Phi}(x) \right] = \mathop{\mathbb{E}}_{z \sim \widehat{\Phi}} \prod_{j \in [n]} \mathds{1}\{z_j = x_j\} = \mathop{\mathbb{E}}_{z \sim \widehat{\Phi}} \mathds{1}\{z = x\} = \widehat{\Phi}(x),
\end{equation}
which is the desired result. This completes the proof.
\end{proof}

The summary of the learning algorithm for low-degree Pauli channel, described so far, is presented in Algorithm~\ref{algo:low_deg_pauli_channels}. We now derive its sample complexity.

\begin{algorithm}[H]
    \caption{Learning low-degree Pauli channels through unentangled measurements} \label{algo:low_deg_pauli_channels}
    \textbf{Input}: Budget $T=O(9^dn^{2d})$, access to Pauli channel $\Phi$, parameter $d$
    \begin{algorithmic}[1]
        \State Initialize $\widetilde{\Phi}(x) = 0$ for each $x \in S$ with $S = \{x : x \in \{0,1,2,3\}^n, \, |x|\leq d\}$
        \For{sample $t=1,\ldots,T$}
            \State Randomly sample string $s \sim \{1,2,3\}^n$
            \State Prepare state $\rho_s = \otimes_{j=1}^s \ket{\chi_{+1}^{s_j}}\bra{\chi_{+1}^{s_j}}$
            \State Query Pauli channel $\Phi$ with $\rho_s$ to obtain $\Phi(\rho_s)$
            \For{qubit $k=1,\ldots,n$}
                \State Measure qubit $k$ of state $\Phi(\rho_s)$ in Pauli basis $\sigma_{s_k}$
                \State Save measurement $r^{(t)}_k \in \{0,1\}$
            \EndFor
        \EndFor 
        \State Compute estimates $\widetilde{\Phi}(x)$ using Eq.~\ref{eq:emp_estimate_phix_pauli_channel} $\forall x \in S$
    \end{algorithmic}
    \textbf{Output}: $\{\widetilde{\Phi}(x)\}_{x \in S}$
\end{algorithm}

\begin{proof}[ of~\cref{theo:learnPauliChannels}]
We learn the channel $\Phi$ by estimating the Fourier coefficients $\widehat{\Phi}(x)$ for each $x \in S$ with $S = \{x : x \in \{0,1,2,3\}^n, \, |x|\leq d\}$. Suppose that we use $T$ measurement outcomes from the learning algorithm in Algorithm~\ref{algo:low_deg_pauli_channels}, then we can estimate $\widehat{\Phi}(x)$ using the empirical estimator from Eq.~\ref{eq:emp_estimate_phix_pauli_channel}. Combining with Prop.~\ref{prop:consistency_emp_est_pauli_channels}, we have by Hoeffding's inequality that
\begin{equation}
    \Pr[|\widetilde{\Phi}(x) - \widehat{\Phi}(x)| \geq a] \leq 2 \exp(-2 a^2 T).
\end{equation}
To learn the channel $\Phi$ within $\eps$-diamond norm, it is enough to learn the probability distribution $\widehat{\Phi} = \{\widehat{\Phi}(x)\}_x$ in total variation distance of $\eps/2$. We require
\begin{equation}
    \frac{1}{2}\|\widetilde{\Phi} - \widehat{\Phi}\|_1 = \frac{1}{2} \sum_{x \in S} |\widetilde{\Phi}(x) - \widehat{\Phi}(x)| \leq \frac{a}{2} 3^d n^d \leq \frac{\eps}{2}.
\end{equation}
To ensure we learn $\Phi$ with probability $1 - \delta$, we let $T = O\left(\frac{9^d n^{2d}}{\eps^2}\log \frac{n}{\delta} \right)$ after applying the union~bound.
\end{proof}
\end{document}